\documentclass[a4paper,11pt]{article}

\usepackage{amsmath,amsthm}
\usepackage{amssymb}
\usepackage{inconsolata}
\usepackage{blkarray}
\usepackage{arydshln}
\usepackage[UKenglish]{babel}
\usepackage{microtype}
\usepackage{a4wide}
\usepackage{thm-restate}
\usepackage{hyperref} 
\usepackage[nameinlink]{cleveref}
\usepackage[svgnames]{xcolor} 
\usepackage{bm}

\hypersetup{colorlinks=true,citecolor=RoyalBlue,linkcolor=RoyalBlue,urlcolor=RoyalBlue}

\newcommand{\NP}{NP}

\DeclareMathOperator{\ar}{ar}
\DeclareMathOperator{\dist}{dist}
\DeclareMathOperator{\diam}{diam}
\DeclareMathOperator{\minmax}{minmax}
\DeclareMathOperator{\CSP}{CSP}
\DeclareMathOperator{\PCSP}{PCSP}
\DeclareMathOperator{\id}{id}
\DeclareMathOperator{\Pol}{Pol}
\DeclareMathOperator{\range}{range}

\DeclareMathOperator{\supp}{supp}

\newcommand{\fr}{\ensuremath{f^\star}}
\newcommand{\M}{\mathcal{M}}
\newcommand{\N}{\mathcal{N}}
\newcommand{\inn}[2]{\mathbf{#1}\textbf{-}\mathbf{in}\textbf{-}\mathbf{#2}}
\newcommand{\NAE}{\mathbf{NAE}}
\newcommand{\BLP}{\ensuremath{\operatorname{BLP}}}
\newcommand{\AIP}{\ensuremath{\operatorname{AIP}}}
\newcommand{\CLAP}{\ensuremath{\operatorname{CLAP}}}
\newcommand{\BLPAIP}{\ensuremath{\operatorname{BLP+AIP}}}

\newcommand{\A}{\mathbf{A}}
\newcommand{\B}{\mathbf{B}}
\newcommand{\E}{\mathbf{E}}
\newcommand{\CC}{\mathbf{C}}
\newcommand{\G}{\mathbf{G}}
\newcommand{\K}{\mathbf{K}}
\newcommand{\X}{\mathbf{X}}
\newcommand{\Eqn}[2]{\mathbf{Eqn}_{#1,#2}}
\newcommand{\trip}{\Gamma}

\theoremstyle{plain}
\newtheorem{theorem}{Theorem}
\newtheorem{lemma}[theorem]{Lemma}
\newtheorem*{lemma*}{Lemma}
\newtheorem{proposition}[theorem]{Proposition}

\newtheorem*{proposition*}{Proposition}

\newtheorem*{corollary*}{Corollary}
\newtheorem{conjecture}[theorem]{Conjecture}
\theoremstyle{definition}
\newtheorem{definition}[theorem]{Definition}
\newtheorem{remark}[theorem]{Remark}
\newtheorem{example}[theorem]{Example}
\newtheorem{problem}[theorem]{Problem}

\begin{document}

\title{On the complexity of symmetric vs.~functional PCSPs\thanks{An extended
abstract of part of this work (with weaker results) appeared in the Proceedings of LICS
2023~\cite{nz23:lics}. This work was supported by UKRI EP/X024431/1 and a Clarendon Fund Scholarship. For the purpose of Open Access, the authors have applied a CC BY public copyright licence to any Author Accepted Manuscript version arising from this submission. All data is provided in full in the results section of this paper.}}

\author{Tamio-Vesa Nakajima\\
University of Oxford\\
\texttt{tamio-vesa.nakajima@cs.ox.ac.uk}
\and
Stanislav \v{Z}ivn\'y\\
University of Oxford\\
\texttt{standa.zivny@cs.ox.ac.uk}
}

\maketitle

\begin{abstract}

The complexity of the \emph{promise constraint satisfaction problem} $\PCSP(\A,\B)$ is largely unknown, even for symmetric $\A$ and $\B$, except for the case when $\A$ and $\B$ are Boolean.

First, we establish a dichotomy for $\PCSP(\A, \B)$ where $\A, \B$ are symmetric, $\B$ is \emph{functional} (i.e.~any $r-1$ elements of an $r$-ary tuple uniquely determines the last one), and $(\A, \B)$ satisfies technical conditions we introduce called \emph{dependency} and \emph{additivity}. This result implies a dichotomy for $\PCSP(\A,\B)$ with $\A,\B$ symmetric and $\B$ functional if (i) $\A$ is Boolean, or (ii) $\A$ is a hypergraph of a small uniformity, or (iii) $\A$ has a relation $R^\A$ of arity at least 3 such that the hypergraph diameter of $(A, R^\A)$ is at most 1.

Second, we show that for $\PCSP(\A, \B)$, where $\A$ and $\B$ contain a single relation, $\A$ satisfies a technical condition called \emph{balancedness}, and $\B$ is arbitrary, the combined \emph{basic linear programming} relaxation ($\BLP$) and the \emph{affine integer programming} relaxation ($\AIP$) is no more powerful than the (in general strictly weaker) $\AIP$ relaxation. Balanced $\A$ include symmetric $\A$ or, more generally, $\A$ preserved by a transitive permutation group.

\end{abstract}

\section{Introduction}\label{sec:intro}
\emph{Promise constraint satisfaction problems} (PCSPs) are a generalisation of constraint satisfaction problems (CSPs) that allow for capturing many more computational problems~\cite{AGH17,BG21:sicomp,BBKO21}.

A canonical example of a CSP is the 3-colouring problem: Given a graph $\G$, is
it 3-colourable? This can be cast as a CSP.\@ Let $\K_k$ denote a clique
on $k$ vertices. Then $\CSP(\K_3)$, the constraint satisfaction problem with the
template $\K_3$, is the following computational problem (in the decision 
version): Given a graph $\G$, say \textsc{Yes} if there is a homomorphism from
$\G$ to $\K_3$ (indicated by $\G\to\K_3$) and say \textsc{No} otherwise
(indicated by $\G\not\to\K_3)$. Here a graph homomorphism is an edge preserving
map~\cite{Hell90:h-coloring}. As graph homomorphisms from $\G$ to $\K_3$ are
$3$-colourings of $\G$, $\CSP(\K_3)$ is the $3$-colouring problem.

Another example of a CSP is 1-in-3-SAT:\@ Given a positive 3-CNF
formula, is there an assignment that satisfies exactly one literal in each
clause? This is $\CSP(\inn{1}{3})$, where
\[
\inn{1}{3}=(\{0,1\};\{(1,0,0),(0,1,0),(0,0,1)\}).
\]
Yet another example is NAE-3-SAT:\@ 
Given a positive 3-CNF formula, is there an assignment that satisfies one or two
literals in each clause? This is $\CSP(\NAE)$, where
\[\NAE=(\{(0,1);{\{0,1\}}^3\setminus\{(0,0,0),(1,1,1)\}).\]

A canonical example of a PCSP is the approximate graph colouring
problem~\cite{GJ76}: Fix $k\leq\ell$. Given a graph $\G$, determine whether $\G$ is $k$-colourable or
not even $\ell$-colourable. (If neither condition holds then the algorithm can output anything; also, note that if $k = \ell$ this is just $k$-colouring.)
This is the same as the PCSP over cliques; i.e., 
$\PCSP(\K_k,\K_\ell)$ is the following computational
problem (in the decision version): Given a graph $\G$, say \textsc{Yes} if
$\G\to\K_k$ and say \textsc{No} if $\G\not\to\K_\ell$. In the search version,
one is given a $k$-colourable graph $\G$ and the task is to find an
$\ell$-colouring of $\G$ (which necessarily exists by the promise and the fact
that $k\leq\ell$).

Another example of a PCSP is $\PCSP(\inn{1}{3},\NAE)$, identified in~\cite{BG21:sicomp}:
Given a satisfiable instance $\X$ of $\CSP(\inn{1}{3})$,
can one find a solution if $\X$ is seen as an instance of $\CSP(\NAE)$? I.e., can
one find a solution that satisfies one or two literals in each clause given the
promise that a solution that satisfies exactly one literal in each clause
exists?
Although both $\CSP(\inn{1}{3})$  and $\CSP(\NAE)$ are \NP-complete, Brakensiek
and Guruswami showed in~\cite{BG21:sicomp} that $\PCSP(\inn{1}{3},\NAE)$ is solvable in polynomial time and in particular it is solved by the so-called affine integer programming relaxation ($\AIP$), whose power was characterised in~\cite{BBKO21}.

More generally, one fixes two relational structures $\A$ and $\B$ with
$\A\to\B$. The $\PCSP(\A,\B)$ is then, in the decision version, the computational problem of distinguishing (input) relational structures $\X$ with $\X\to\A$ from those with $\X\not\to\B$.
In the search version, $\PCSP(\A,\B)$ is the problem of finding a homomorphism from an input structure $\X$ to $\B$ given that one is promised that $\X\to\A$. 
One can think of $\PCSP(\A,\B)$ as an approximation version of $\CSP(\A)$ on satisfiable instances. Another way is to think of $\PCSP(\A,\B)$ as $\CSP(\B)$ with restricted inputs. 
We refer the reader to~\cite{KO22:survey} for a very recent survey on PCSPs.

\medskip

For CSPs, a dichotomy conjecture of Feder and Vardi~\cite{Feder98:monotone} was
resolved independently by Bulatov~\cite{Bulatov17:focs} and
Zhuk~\cite{Zhuk20:jacm} via the so-called algebraic
approach~\cite{Jeavons98:algebraic,Bulatov05:classifying}: For every fixed
finite $\A$, $\CSP(\A)$ is either
solvable in polynomial time or $\CSP(\A)$ is \NP-complete. 

For PCSPs, even the case of graphs and structures on Boolean domains is widely open; these two were established for CSPs a long time ago~\cite{Hell90:h-coloring, Schaefer78:stoc} and constituted important evidence for conjecturing a dichotomy. 
Following the important work of Barto et al.~\cite{BBKO21} on extending the
algebraic framework from the realms of CSPs to the world of PCSPs, there have
been several recent works on complexity classifications of fragments of
PCSPs~\cite{Ficak19:icalp,GS20:icalp,BG21:sicomp,BWZ21,Barto21:stacs,AB21,BG21:talg,bz22:ic,NZ22:toct,KOWZ23},
hardness conditions~\cite{BBKO21,BWZ21,Barto22:soda,Wrochna22}, and power of
algorithms~\cite{BBKO21,BGWZ,Atserias22:soda,cz23sicomp:clap}. Nevertheless, 
a classification of more concrete fragments of PCSPs is needed for making
progress with the general theory, such as finding hardness and tractability
criteria, as well as with resolving longstanding open questions, such as approximate graph colouring.

Brakensiek and Guruswami classified $\PCSP(\A,\B)$ for all Boolean symmetric structures $\A$ and $\B$ with disequalities~\cite{BG21:sicomp}. Ficak, Kozik, Ol\v{s}{\'a}k, and Stankiewicz generalised this result by classifying $\PCSP(\A,\B)$ for all Boolean symmetric structures $\A$ and $\B$~\cite{Ficak19:icalp}.

Barto, Battistelli, and
Berg~\cite{Barto21:stacs} studied symmetric PCSPs on non-Boolean domains and in particular PCSPs of the form $\PCSP(\inn{1}{3},\B)$, 
where $\B$ contains a single ternary relation over the domain $\{0,1,\ldots,d-1\}$.
For $d=2$, a complete classification $\PCSP(\inn{1}{3},\B)$ is
known~\cite{BG21:sicomp,Ficak19:icalp}. For $d=3$, Barto et
al.~\cite{Barto21:stacs} managed to classify all but one structure $\B$. The
remaining open case of ``linearly ordered colouring'' inspired further investigation in~\cite{NZ22:toct}.
For $d=4$, Barto et al.~\cite{Barto21:stacs} obtained partial results. In
particular,  for certain structures
$\B$ they managed to rule out the applicability of  the $\BLPAIP$
algorithm from~\cite{BGWZ}. The significance of $\BLPAIP$ here is that it is the
strongest known algorithm for PCSPs for which 
a characterisation of its power is
known both in terms of a minion\footnote{We will only use polymorphism minions~\cite{BBKO21} in this paper. More general minions capture not only the power of $\BLPAIP$~\cite{BGWZ} but also the power of the $\CLAP$ algorithm~\cite{cz23sicomp:clap}.} and also in terms of polymorphism identities, cf.~\cite{BGWZ} for details. Furthermore, $\BLPAIP$ solves all currently known tractable Boolean PCSPs~\cite{Ficak19:icalp,BGWZ}.
This suggests that those cases are \NP-hard (or
new algorithmic techniques are needed).

\paragraph{Contributions.}
We continue the work from~\cite{BG21:sicomp,Ficak19:icalp} and~\cite{Barto21:stacs} and
focus on promise constraint satisfaction
problems of the form $\PCSP(\A,\B)$,
where $\A$ is symmetric and $\B$ is over an arbitrary finite domain.\footnote{All structures in this article can be assumed to be finite unless they are explicitly stated to be infinite.}
Since the template $\A$ is symmetric, we can assume without loss of generality
that $\B$ is symmetric, as observed in~\cite{Barto21:stacs} and in~\cite{bz22:ic}.\footnote{In detail, for any symmetric $\A$ and (not
necessarily symmetric) $\B$ with $\A\to\B$, there is a symmetric $\B'$ with
$\A\to\B'$ such that $\PCSP(\A,\B)$ and $\PCSP(\A,\B')$ are polynomial-time
equivalent~\cite{Barto21:stacs,bz22:ic}. This $\B'$ is the largest symmetric substructure of $\B$. Observe that a functional structure has functional substructures, so if $\B$ is functional then $\B'$ remains functional.}

As our main result, we establish the following result. A structure $\B$ is
called \emph{functional} if, for any relation $R^\B$ in $\B$ of, say, arity $r$,
and any tuple $x\in R^\B$, any $r-1$ elements of $x$ determine the last element. In detail,
$(x_1,\ldots,x_{r-1},y),(x_1,\ldots,x_{r-1},z)\in R^\B$
implies $y=z$, and similarly for the other $r-1$ positions.\footnote{Note that
for symmetric $\B$ the requirement ``for the other $r-1$ positions'' is satisfied automatically.}
The notions of dependency and additivity will be defined in
\Cref{sec:additivityDependency}: for the moment, consider them simply as
technical conditions on $(\A, \B)$.
Finite tractability is defined in
\Cref{sec:prelims} --- informally, a template is finitely tractable if it can be solved by treating the instance as though it were an instance of a tractable finite-domain CSP.

\begin{restatable}{theorem}{thmdichotomy}%
\label{thm:dichotomy}
Let $\A$ be a symmetric structure and $\B$ be a functional structure such that $\A \to \B$. 
Assume that $(\A, \B)$ is dependent and additive. 
Then, either $\PCSP(\A, \B)$ is solvable in polynomial time by $\AIP$ and is finitely tractable, or $\PCSP(\A,\B)$ is \NP-hard.
\end{restatable}

Our main motivation for studying PCSPs with functional structures is the fact (mentioned above)
that more complexity classifications of PCSP fragments are needed to make progress with the general theory of PCSPs.
Furthermore, functional PCSPs generalise linear equations, an important and
fundamental class of CSPs. Finally,
the topological methods that proved useful for showing the hardness of certain PCSPs (e.g.~$\PCSP(G, K_3)$ for any non-bipartite 3-colourable graph $G$, where $K_3$ is the clique~\cite{KOWZ23}, or approximate 3-vs.-4 linearly ordered colouring~\cite{fnotw:stacs}) seem inapplicable to 
functional PCSPs and thus other methods are required.

\Cref{thm:dichotomy} has the following three corollaries. The first
corollary applies to structures $\A$ that are Boolean.

\begin{restatable}{corollary}{corBoolDich}%
\label{cor:BooleanDichotomy}
Let $\A$ be a Boolean symmetric structure and $\B$ be a functional structure such that $\A \to \B$. 
Then, either $\PCSP(\A, \B)$ is solvable in polynomial time by $\AIP$ and is finitely tractable, or $\PCSP(\A,\B)$ is \NP-hard.
\end{restatable}

\noindent
The second corollary applies to structures $\A$ consisting of a single relation of a small arity.

\begin{restatable}{corollary}{corTernDich}%
\label{cor:TernDich}
Let $\A$ be a symmetric structure and $\B$ be a functional structure such that $\A \to \B$. 
Assume that both $\A$ and $\B$ have exactly one relation of arity at most 4. 
Then, either $\PCSP(\A, \B)$ is solvable in polynomial time, or $\PCSP(\A,\B)$ is \NP-hard.
\end{restatable}

An example of a Boolean symmetric structure $\A$ is $\inn{1}{3}$, and more generally $\inn{q}{r}$.\footnote{$\inn{q}{r}$ is the structure on $\{0,1\}$ with a single (symmetric) relation of arity $r$ containing all $r$-tuples with precisely $q$ $1$s (and $r-q$ $0$s).}
The structure $\inn{1}{3}$ is not only Boolean (and thus is captured by
\Cref{cor:BooleanDichotomy}) but also consists of a single relation of
arity 3 (and thus is also captured by \Cref{cor:TernDich}). In fact, it
satisfies another property, leading to the following generalisation of
$\inn{1}{3}$, which applies to structures $\A$ with a certain connectivity
property.
We will need some notation.
Let $\dist_{R}(x, y)$  be the distance between $x$ and $y$ when viewed as vertices in a hypergraph whose edge relation is $R$; in particular, $\dist_R(x, x) = 0$ and $\dist_R(x, y) = \infty$ if $x$ and $y$ are in different connected components.
Define $\diam(A, R) = \max_{u, v \in A} \dist_R(u, v)$.

\begin{restatable}{corollary}{corDiamDich}%
\label{cor:DiamDich}
Let $\A$ be a symmetric structure and $\B$ be a functional structure such that $\A \to \B$. 
Assume that $\A$ has a relation $R^\A$ of arity at least 3 for which $\diam(A, R^\A) \leq 1$. 
Then, either $\PCSP(\A, \B)$ is solvable in polynomial time by $\AIP$ and is finitely tractable, or $\PCSP(\A,\B)$ is \NP-hard.
\end{restatable}

A hypergraph is called \emph{linear} if no two distinct edges intersect in more
than one vertex. We remark that any \emph{linear} $r$-uniform hypergraph can be
seen as a functional symmetric relational structure with one $r$-ary relation. 

\medskip
Several researchers have informally conjectured that $\PCSP(\inn{1}{3},\B)$
admits a dichotomy. The authors, as well as other researchers, believe that in
fact not only is there a dichotomy but also all tractable cases are solvable by
$\AIP$ (cf.~also \Cref{rem:notbyAIP}).

\begin{conjecture}\label{conj:aip}
For every structure $\B$, either $\PCSP(\inn{1}{3},\B)$ is solvable in polynomial time by $\AIP$, or $\PCSP(\inn{1}{3},\B)$ is \NP-hard.
\end{conjecture}

\Cref{thm:dichotomy} establishes the special case of \Cref{conj:aip} for functional $\B$.
We make further progress towards \Cref{conj:aip} by proving that for any
structure $\A$ with a single \emph{(not necessarily Boolean)} symmetric relation, and any \emph{(not necessarily functional)} structure $\B$ for which $\A \to \B$, $\BLPAIP$ from~\cite{BGWZ}
is no more powerful for $\PCSP(\A,\B)$ than $\AIP$ 
from~\cite{BG21:sicomp}, although in general $\BLPAIP$ is strictly stronger than
$\AIP$~\cite{BGWZ}, already for (non-promise) $\CSP$s with \emph{two} Boolean
symmetric relations, cf.~\Cref{rem:two}. In fact, we establish a more
general result.
We say that a relation $R$ is \emph{balanced} if there exists a matrix $M$ whose
columns are tuples of $R$, where each tuple of $R$ appears as a column (possibly
a multiple times), and
where the rows of $M$ are permutations of each other.
The matrix $M$ below shows that the Boolean $\inn{1}{3}$ relation is balanced:
\[
M =
\begin{pmatrix}
1 & 0 & 0 \\
0 & 1 & 0 \\
0 & 0 & 1 \\
\end{pmatrix}.
\]

\begin{restatable}{theorem}{thmblpaip}%
\label{thm:collapse}
Let $\A$ be any structure with a single relation. If the relation in $\A$ is balanced then, for any $\B$ such that $\A \to \B$, $\BLPAIP$ solves $\PCSP(\A, \B)$ if and only if $\AIP$ solves it.
\end{restatable}

If the (only) relation in $\A$ is \emph{binary} (i.e., a digraph), the condition of balancedness has a natural combinatorial interpretation: A binary relation is balanced if and only if it is the disjoint union of strongly connected components (cf.~\Cref{app:lemmas}). 

\Cref{thm:collapse} implies the following corollary. We say that a relation of arity $r$ is \emph{preserved} by a group of permutations of degree $r$ if and only if permuting any tuple of the relation according to any permutation of the group gives another tuple of the relation.

\begin{restatable}{corollary}{cortransitivegroup}%
\label{cor:transitivegroup}
Suppose that $G$ is a transitive group of permutations, of order $r$. Further, suppose that $\A$ is a relational structure with one relation, of arity $r$, that is preserved by $G$. Then, for any $\A \to \B$, $\BLPAIP$ solves $\PCSP(\A, \B)$ if and only if $\AIP$ does.
\end{restatable}

While \Cref{cor:transitivegroup} is more elegant than \Cref{thm:collapse}, it applies to fewer structures. Indeed, we will show in \Cref{rem:eulerian} that there exist balanced relations that are not preserved by any transitive group.
Examples of relations that are preserved by some transitive group of permutations $G$ include symmetric relations (where $G$ is the symmetric group) or cyclic relations (where $G$ contains all cyclic shifts of appropriate degree).

\section{Preliminaries}%
\label{sec:prelims}

We let $[r]= \{1, \ldots, r\}$. We denote by $2^S$ the
powerset of $S$. 

\paragraph{Structures and PCSPs}
Promise CSPs have been introduced in~\cite{AGH17} and~\cite{BG21:sicomp}. We follow the notation and terminology of~\cite{BBKO21}.

A (relational) \emph{structure} is a tuple $\A=(A; R_1^\A, \ldots, R_t^\A)$, where $R_i^\A \subseteq
A^{\ar(R_i)}$ is a relation of arity $\ar(R_i)$ on a set $A$, called the \emph{domain}. 
A structure $\A$ is called \emph{Boolean} if $A=\{0,1\}$ and is called \emph{symmetric} if $R_i^\A$ is a symmetric relation for each $i \in [t]$; i.e, if
$(x_1,\ldots,x_{\ar(R_i)})\in R_i^\A$ then for every permutation $\pi$ on $[\ar(R_i)]$ we have $(x_{\pi(1)},\ldots,x_{\pi(\ar(R_i))})\in R_i^\A$.
A structure $\A$ is called \emph{functional} if 
$(x_1,\ldots,x_{\ar(R_i)-1},y)\in R_i^\A$
and 
$(x_1,\ldots,x_{\ar(R_i)-1},z)\in R_i^\A$ 
implies $y=z$ for any $i \in [t]$, and that the same hold for all other $r-1$ positions in the tuple. For any $r$-ary functional relation $R \subseteq A^r$, we define a partial map also called $R$ from $A^{r-1}$ to $A$ in the following way: for any $x_1, \ldots, x_{r-1} \in A$, $R(x_1, \ldots, x_{r-1})$ is the unique value $y$ such that $(x_1, \ldots, x_{r-1}, y) \in R$, if it exists; $R(x_1, \ldots, x_{r-1})$ is undefined if no such value exists.

Consider two structures $\A=(A;R_1^\A,\ldots,R_t^\A)$ and $\B=(B;R_1^\B,\ldots,R_t^\B)$ with $t$ relations, where, for each $i \in [t]$, $R_i^\A$ and $R_i^\B$ have the same arity. 
A \emph{homomorphism from $\A$ to $\B$} is a function $h : A \to B$ such that, for any $i \in [t]$, for
each $x=(x_1,\ldots,x_{\ar(R_i)})\in R_i^\A$, we have $h(x)=(h(x_1),\ldots,h(x_{\ar(R_i)}))\in R_i^\B$. We denote the
existence of a homomorphism from $\A$ to $\B$ by $\A \to \B$. 

Let $\A$ and $\B$ be two structures with $\A\to\B$; we call $(\A,\B)$
a (PCSP) \emph{template}. In the \emph{search} version of the \emph{promise constraint
satisfaction problem} (PCSP) with the template $(\A,\B)$, denoted by
$\PCSP(\A,\B)$, the task is: Given a structure $\X$ with the promise that $\X\to
\A$, find a homomorphism from $\X$ to $\B$ (which necessarily exists as
homomorphisms compose). In the \emph{decision} version of $\PCSP(\A,\B)$, the
task is: Given a structure $\X$, output \textsc{Yes} if $\X \to \A$, and output
\textsc{No} if $\X\not\to\B$.%
\footnote{If neither condition holds then the algorithm can output anything. An equivalent way is to define the problem as follows: Given a relational structure $\X$ such that either $X\to\A$ or $\X\not\to\B$, which is the promise, output \textsc{Yes} if $\X\to\A$ and output \textsc{No} if $\X\not\to\B$. The equivalence, up to polynomial-time solvability, of the two definitions relies on the fact that polynomials are time-constructible and so a clock can be run.}

The decision version trivially reduces to the search version. We will use
the decision version in this paper.

We will be interested in the complexity of $\PCSP(\A,\B$), in
particular for symmetric $\A$ and functional $\B$. 
(As discussed in \Cref{sec:intro}, the symmetricity of $\A$ means that we can without loss of generality assume symmetricity of $\B$.)

\paragraph{Operations and polymorphisms}
A function $h:A^n\to B$ is called an \emph{operation} of arity $n$. 
A $(2n+1)$-ary operation $f:A^{2n+1}\to B$ is called \emph{$2$-block-symmetric}
if $f(a_1,\ldots,a_{2n+1})=f(a_{\pi(1)},\ldots,\allowbreak a_{\pi(2n+1)})$ for
every $a_1,\ldots,a_{2n+1}\in A$ and every permutation $\pi$ on $[2n+1]$
that preserves parity; i.e, $\pi$ maps odd values to odd values and even values to even values.

A $(2n+1)$-ary operation $f:A^{2n+1}\to B$ is called \emph{alternating} if
it is 2-block-symmetric, and furthermore
$f(a_1,\ldots,a_{2n-1},a,a)=f(a_1,\ldots,a_{2n-1},a',a')$
for every
$a_1,\ldots,a_{2n-1},a,a'\in A$.

Consider structures $\A, \B$ with $t$ relations with the same arities. We call $h:A^n\to B$ a \emph{polymorphism} of $(\A,\B)$ if 
the following holds for any relation $R = R_i$, $i \in [t]$, of arity $r=\ar(R)$. For any $x^1,\ldots,x^r\in A^n$, where $x^i = (x_1^i, \ldots, x_n^i)$, with
$(x^1_i,\ldots,x^r_i) \in R^\A$ for every $1 \leq i \leq n$, we
have $(h(x^1),\ldots,h(x^r))\in R^\B$. One can visualise this as an $(r\times n)$
matrix whose rows are the tuples $x^1,\ldots,x^r$. The requirement is that
if every column of the matrix is in $R^\A$ then the application of
$h$ on the rows of the matrix results in a tuple from $R^\B$. 
We denote by $\Pol^{(n)}(\A,\B)$ the set
of $n$-ary polymorphisms of $(\A,\B)$ and by $\Pol(\A,\B)$
the set of all polymorphisms of $(\A,\B)$.

\paragraph{Relaxations}
There are two standard polynomial-time solvable relaxations for PCSPs, the \emph{basic linear
programming} relaxation ($\BLP$) and the \emph{affine integer programming}
relaxation ($\AIP$)~\cite{BG21:sicomp}. The $\AIP$ solves most tractable
PCSPs studied in this paper, with the exception of cases covered in
\Cref{cor:TernDich} (cf.~also \Cref{rem:notbyAIP}). 
There is also a combination of the two, called
$\BLPAIP$~\cite{BGWZ}, that is provably stronger than both $\BLP$ and 
$\AIP$.
We will show that for certain PCSPs, this is
not the case (cf.~\Cref{thm:collapse}). 
The precise definitions of the relaxations are not important for us as we will only need the notion of solvability of PCSPs by these relaxations and characterisations of the power of the relaxations;  we refer
the reader to~\cite{BG21:sicomp,BBKO21,BGWZ} for details.
Let $\X$ be an instance of $\PCSP(\A,\B)$. It follows from the
definitions of the relaxations that if $\X\to\A$ then both $\AIP$ and
$\BLPAIP$ accept~\cite{BG21:sicomp,BBKO21}. We say that $\AIP$ ($\BLPAIP$, respectively)
\emph{solves} $\PCSP(\A,\B)$ if, for every $\X$ with $\X\not\to\B$,
$\AIP$ ($\BLPAIP$, respectively) rejects.

The power of $\AIP$ and $\BLPAIP$ for PCSPs is characterised by the following
results.

\begin{theorem}[\cite{BBKO21}]\label{thm:aip}
  $\PCSP(\A,\B)$ is solved by $\AIP$ if and only if $\Pol(\A,\B)$ contains alternating operations of all odd arities.
\end{theorem}

\begin{theorem}[\cite{BGWZ}]\label{thm:blpaip}
  $\PCSP(\A,\B)$ is solved by $\BLPAIP$ if and only if $\Pol(\A,\B)$ contains 2-block-symmetric operations of all odd arities.
\end{theorem}

We now define the notion of \emph{finite tractability}~\cite{BBKO21,AB21}. We say that $\PCSP(\A, \B)$ is finitely tractable if $\A \to \mathbf{E} \to \B$ for some finite structure $\mathbf{E}$ and $\CSP(\mathbf{E})$ is tractable.
For a group $G$, we use the standard notation $H\lhd G$ to indicate that $H$ is a normal subgroup of $G$.

\begin{lemma}\label{lem:sandwich}
Suppose $\A \to \mathbf{E} \to \B$, where $E = G$ for some finite Abelian group $(G,
  +)$, and each relation of $\mathbf{E}$ is either of the form (i) $c + H$ for
  some $r \in \mathbb{N}$, $c \in G^r$ and $H \lhd G^r$, or (ii) empty. 
  Then, $\PCSP(\A, \B)$ is solvable in polynomial time by $\AIP$ and is finitely tractable.
\end{lemma}
\begin{proof}
The following alternating operation is a polymorphism of $\mathbf{E}$
\[
f(x_1, y_1, \ldots, y_{k}, x_{k+1}) = \sum_{i = 1}^{k + 1} x_i - \sum_{i = 1}^k y_i.
\]
Consider a relation $R^\mathbf{E}$ of $\mathbf{E}$, of the form $c + H$. Consider a matrix of inputs whose columns are $x_1, y_1, \ldots, y_k, x_{k+1} \in R^\mathbf{E}$. In other words, $x_i \in c + H$ and $y_i \in c + H$ for each $x_i, y_i$. Note that the column that results from applying $f$ to the rows of this matrix is just
\[
x_1 - y_1 + \cdots - y_k + x_{k+1} \in (c + H) - (c + H) + \cdots - (c+H) + (c+H) \subseteq c + H
\]
Thus $f$ is an alternating polymorphism of $\E$. It follows that $\CSP(\E)$ is solved by $\AIP$, from whence it follows that $\PCSP(\A, \B)$ is finitely tractable and solved by $\AIP$.
\end{proof}

\paragraph{Minions} 
We will use the theory of minions from~\cite{BBKO21}.
Let $\M$ be a set, where each element $f \in \M$ is assigned an arity $\ar(f)$.
We write $\M^{(n)} = \{ f \in \M \mid \ar(f) = n \}$. Further, let $\M$ be
endowed with, for each $\pi : [n] \to [m]$, a (so-called minor) map $f \mapsto f^\pi : \M^{(n)}
\to \M^{(m)}$ such that, for $\pi : [n] \to [m]$ and $ \sigma : [m] \to [k]$,
and any $f \in \M^{(n)}$ we have ${(f^\pi)}^\sigma = f^{\sigma \circ \pi} \in \M^{(k)}$, and
$f^{\textrm{id}} = f$. Then, $\M$ is called a \emph{minion}.\footnote{A minion is a functor from the skeleton of the category of finite sets to the category of sets.} We often write $f \xrightarrow{\pi} g$ instead of $g = f^\pi$.

Consider two minions $\M, \N$; a \emph{minion homomorphism} is a map $\xi : \M
\to \N$ such that, for any $f \in \M^{(n)}$ and $\pi : [n] \to [m]$, we have
that ${\xi(f)}^\pi = \xi(f^\pi)$.\footnote{Minion homomorphisms are natural transformations.} If such a minion homomorphism exists, we write $\M \to \N$.

Given an $n$-ary operation $f:A^n\to B$
and a map $\pi : [n] \to [m]$, an $m$-ary operation $g:A^m\to B$ is called  a \emph{minor} of $f$ given by the map $\pi$ if \[
    g(x_1, \ldots, x_m) = f(x_{\pi(1)}, \ldots, x_{\pi(n)}).
\]
The polymorphisms $\Pol(\A, \B)$ thus form a minion, where $f^\pi$ is given by the minor of $f$ at $\pi$.

The main hardness theorem that we will use is the
following.\footnote{In~\cite{BBKO21}, this is Theorem~5.21 together with
Lemma~5.11, as detailed after the proof of Theorem~5.21 therein.}

\begin{theorem}[{\cite{BBKO21}}]\label{thm:chain}
Fix constants $m$ and $C$. Take any template $(\A, \B)$. Suppose $\Pol(\A, \B) = \bigcup_{i = 1}^m \M_i$ such that for every $i \in [m]$ there exists a map $I_i$ that takes $f \in \M_i$ to a subset of $[\ar(f)]$ of size at most $C$ such that the following holds: For any $f, g
  \in \M_i, \pi : [\ar(f)] \to [\ar(g)]$ such that $g = f^\pi$ we have that $I_i(g) \cap \pi(I_i(f)) \neq
  \emptyset$. Then, $\PCSP(\A, \B)$ is \NP-hard.
\end{theorem}

\section{Additivity and dependency}%
\label{sec:additivityDependency}

In this section we will define two new concepts for a template $(\A, \B)$,
additivity and dependency. Let $f\in\Pol(\A,\B)$ be a polymorphism of $(\A,\B)$.
Intuitively,  additivity constrains the value of $f$ evaluated at two elements
and dependency ensures that $f$ is determined by such evaluations.  
The end goal of this section is to prove the following two theorems.

\begin{theorem}\label{thm:diamAdditivityDependency}
Suppose $\A$ has a symmetric relation $R^\A$ of arity at least 3 with $\diam(A, R^\A) \leq 1$. Then, for any functional $\B$ such that $\A \to \B$, $(\A, \B)$ is additive and dependent.
\end{theorem}

\begin{theorem}\label{thm:4AdditivityDependency}
Suppose $\A$ has a symmetric relation $R^\A$ of arity 3 or 4 that is connected when viewed as the edge relation of a hypergraph on vertices $A$. Then, for any functional $\B$ such that $\A \to \B$, $(\A, \B)$ is additive and dependent.
\end{theorem}

Throughout this section, we will assume implicitly that $\A, \B$ are symmetric.

\subsection{Additivity}

Consider a polymorphism $f \in \Pol^{(n)}(\A, \B)$. For any $( i, j ) \in A^2$ (including the case where $i = j$), we can define a function $f_{ij}$ derived from $f$ in the following way: $f_{ij} : 2^{[n]} \to B$ is a function where 
\[
f_{ij}(S) = f(x_1, \ldots, x_n)
\]
where $x_k = j$ if $k \in S$ and $x_k = i$ otherwise. In other words, $f_{ij}(S)$ is $f$ evaluated at the characteristic vector of $S$, where $j$ indicates membership in $S$ and $i$ indicates non-membership. (Equivalently, let $\pi_S : [n] \to [2]$ be given by $\pi_S(k) = 2$ if $k \in S$ and $\pi_S(k) = 1$ otherwise. Then, $f_{ij}(S) = f^{\pi_S}(i, j)$.)

We define $f^p : 2^{[n]} \to B^{A^2}$ to be the function\footnote{The superscript ``\textit{p}'' indicates the word ``pair''.}
\[
f^p(S)(i, j) = f_{ij}(S).
\]

We will be interested in templates $(\A, \B)$ that have the following property.

\begin{definition}\label{def:additivity}
    Consider a template $(\A, \B)$ with $\A$ symmetric and $\B$ functional. We say that $(\A, \B)$ is \emph{additive} if there exists an operator $+ : B^{A^2} \times B^{A^2} \to B^{A^2}$ such that, for any $f \in \Pol^{(n)}(\A, \B)$ and disjoint $S, T \subseteq [n]$ we have
    \[
        f^p(S) + f^p(T) = f^p(S \cup T).
    \]
\end{definition}

\begin{lemma}\label{lem:add-sub}
    If $(\A, \B)$ is additive, there exists an operator $- : B^{A^2} \times B^{A^2} \to B^{A^2}$ such that, for any $f \in \Pol^{(n)}(\A, \B)$ and $S \subseteq T \subseteq [n]$, we have
    \[
    f^p(T) - f^p(S) = f^p(T \setminus S).
    \]
\end{lemma}
\begin{proof}
While $f^p(S)$ has been written as a function from $A^2$ to $B$ above, we can also see it as an $|A| \times |A|$ matrix of elements of $B$. Thus we can take the transpose of this matrix, denoted by the superscript $T$ below.
Observe that ${(f^p(S))}^T(i,j) = f_{ji}(S) = f_{ij}(\overline{S}) = f^p(\overline{S})(i, j)$, where $\overline{S}$ denotes the complement of $S$. In other words, ${(f^p(S))}^T = f^p(\overline{S})$.

    Set $x - y = {(x^T + y)}^T$. Then for $S\subseteq T \subseteq [n]$,
    \[
    f^p(T) - f^p(S) = {({f^p(T)}^T + f^p(S))}^T = f^p(\overline{\overline{T} \cup S}) = f^p(T \setminus S).\qedhere
    \]
\end{proof}

\begin{lemma}\label{lem:genZero}
Suppose $(\A, \B)$ is additive. Consider a polymorphism $f \in \Pol^{(n)}(\A, \B)$. Consider any family of disjoint sets $\mathcal{A} \subseteq 2^{[n]}$, containing at least $|B|^{|A|^2}$ sets.
Then some nonempty subset $\mathcal{B} \subseteq \mathcal{A}$ exists such that $f^p(\bigcup \mathcal{B}) = f^p(\emptyset)$.
\end{lemma}
The approach used to prove this is analogous to the following well known exercise (first set out by V\'{a}zsonyi and Sved, according to Erd\"{o}s~\cite{thebook}):  
Prove that any sequence of $n$ integers has a subsequence whose sum is divisible by $n$.
\begin{proof}
$\mathcal{A}$ contains at least $|B|^{|A|^2} \geq |\range(f^p)|$
  different sets. Let $A_1, \ldots, A_{|\range(f^p)|}$ be some of these sets.
  Define $B_i = \bigcup_{j \leq i} A_j$ for $0 \leq i \leq |\range(f^p)|$; note
  that $B_0 = \emptyset$. By the pigeonhole principle there exists $0\leq i < j
  \leq |\range(f^p)|$ such that $f^p(B_i) = f^p(B_j)$. Then, using
  \Cref{lem:add-sub}, $f^p(B_j \setminus B_i) = f^p(B_j) - f^p(B_i) = f^p(B_i) - f^p(B_i) = f^p(B_i \setminus B_i) = f^p(\emptyset)$. Thus $\mathcal{B} = \{ A_{i+1}, \ldots, A_j \}$ is the required family of sets.
\end{proof}

\begin{lemma}\label{lem:genSmallSet}
Suppose $(\A, \B)$ is additive. Consider a polymorphism $f \in \Pol^{(n)}(\A, \B)$. Consider any $S \subseteq [n]$. There exists $T \subseteq S$ of size at most $|B|^{|A|^2}$ such that $f^p(S) = f^p(T)$.
\end{lemma}
\begin{proof}
  Suppose this is not the case, and suppose that $S$ is a minimal counterexample
  (with respect to inclusion) to this claim. Clearly $|S| > |B|^{|A|^2}$, or else taking $T = S$ shows that $S$ is no counterexample at all. Thus, apply \Cref{lem:genZero} to the family $\{ \{ x \} \mid x \in S \}$ to find that some nonempty subset $U \subseteq S$ exists such that $f^p(U) = f^p(\emptyset)$. But now, take $S' = S \setminus U \subseteq S$, and note that $f^p(S') = f^p(S \setminus U) = f^p(S) - f^p(U) = f^p(S) - f^p(\emptyset) = f^p(S \setminus \emptyset) =f^p(S)$. By the minimality of $S$, $S'$ has a subset $T$ of size at most $|B|^{|A|^2}$ such that $f^p(T) = f^p(S') = f^p(S)$, which contradicts the fact that $S$ is a counterexample.
\end{proof}

\subsection{Dependency}

In the sequel, we will use ``tuple builder'' notation. In other words, for an indexing set $I$, we write
\[
(f(i) \mid i \in I),
\]
to denote the tuple, indexed by $I$, with elements $f(i)$ for every $i \in I$. For instance,
\[
(x^2 \mid x \in [5]) = (1, 4, 9, 16, 25).
\]
For a more complicated example, we have the tuple
\[
(f(x) \mid x \in [5], f : [5] \to [6]),
\]
where the elements are $f(x)$, indexed by $x \in [5]$, and by functions $f : [5] \to [6]$.

This notation is justified since an $n$-ary tuple is just a function from $[n]$ to some base set; our notation  $(f(i) \mid i \in I)$ ``builds'' the tuple equivalent to the function $f$.

\begin{definition}
    We call $(\A, \B)$ \emph{dependent} if there exists a map $h$ such that, for every $f \in \Pol^{(n)}(\A, \B)$, we have that
    \[
    f(x_1, \ldots, x_n) = h( f( \alpha(x_1), \ldots, \alpha(x_n)) \mid \alpha : A \to S, S \in A^{\leq 2}),
    \]
    where $A^{\leq 2}$ is the set of nonempty subsets of $A$ of size at most 2.
\end{definition}

Informally, $(\A, \B)$ is dependent if, for every $f$, we can deduce the value of $f(x_1, \ldots, x_n)$ from the values $f(\alpha(x_1), \ldots, \alpha(x_n))$ for $\alpha$ a function that maps $A$ to a subset of $A$ of size at most 2. One can think of $\alpha$ as ``flattening'' the values of $x_1, \ldots, x_n$ to only two values. This deduction is encoded by $h$, and must be independent of the choice of $f$.

Any polymorphism $f\in\Pol^{(n)}(\A,\B)$ is just a function $f : A^n \to B$. Without
loss of generality, identify $A$ with the set $[a]$ from now on. Note that a
tuple from $A^n = {[a]}^n$ can be seen as a partition of $[n]$ into $a$ parts $S_1,
\ldots, S_a$: $S_i$ is the set of coordinates in the tuple set to $i$. We will
thus denote by $a^{[n]}$ both the set tuples and the set of such partitions.
Thus we can, for example, evaluate $f$ at a partition $S_1, \ldots, S_a$ of $[n]$ and get $f(S_1, \ldots, S_a)$.

\begin{definition}
    For any polymorphism $f \in \Pol^{(n)}(\A, \B)$, we define $\fr : a^{[n]} \to {(B^{A^2})}^a$ by
    \[
    \fr(S_1, \ldots, S_a) = (f^p(S_1), \ldots, f^p(S_a)).
    \]
\end{definition}

The usefulness of all the concepts introduced so far comes from the following fact. 

\begin{lemma}\label{lem:hExists}
Suppose $(\A, \B)$ is dependent and additive. Then there exists a function $h : {(B^{A^2})}^a \to B$ such that for any polymorphism $f \in \Pol^{(n)}(\A, \B)$, we have $f = h \circ \fr$.
\end{lemma}
In the proof we will freely mix the two notations for functions of form ${[a]}^n \to B$. If the input to such functions is written with capital letters it uses the ``families of disjoint sets'' notation, and if it is written with lower case letters it uses the tuple notation.
\begin{proof}
    We claim that for any $\alpha : A \to S, S \in A^{\leq 2}$ we can deduce $f(\alpha(x_1), \ldots, \alpha(x_n))$ from the values of $\fr(x_1, \ldots, x_n)$ in a manner independent of $f$. If this is the case, then by running this deduction for each $\alpha$, there exists some map $h'$ such that for every $f$ and $x_1, \ldots, x_n$ we have
    \[
    h'(\fr(x_1, \ldots, x_n)) = ( f(\alpha(x_1), \ldots, \alpha(x_n)) \mid \alpha : A \to S, S \in A^{\leq 2}).
    \]
    Hence by the dependency of $(\A, \B)$ there exists a map $h''$ such that $f = h'' \circ h' \circ \fr$, which implies our conclusion.

    First suppose $S = \{s\}$.
    Then we want to find $f(s, \ldots, s)$. But this value is just $f_{ss}(X)$ for any $X$, and thus is included in every $f^p(S_i)$ that is within $\fr(x_1, \ldots, x_n)$.
    
    Now suppose $S = \{b, c\}$, and that $A = B \cup C$ such that $\alpha$ maps $B$ to $b$ and $C$ to $c$. In this case we observe that $(\alpha(x_1), \ldots, \alpha(x_n))$ contains a $c$ at all places where $(x_1, \ldots, x_n)$ contained an element in $C$, and contains a $b$ elsewhere. Suppose $S_1, \ldots, S_a$ are sets such that $S_i$ contains all the indices where $(x_1, \ldots, x_n)$ is equal to $i$. 
    Thus, $f(\alpha(x_1), \ldots, \alpha(x_n)) = f_{bc}\left(\bigcup_{i \in C} S_i\right)$. This value is an element of the tuple $f^p\left(\bigcup_{i \in C} S_i\right)$, so it is sufficient to show that we can deduce this latter value from $\fr(x_1, \ldots, x_n) = \fr(S_1, \ldots, S_a)$. But
    \[
    f^p\left(\bigcup_{i \in C}S_i\right) = \sum_{i \in C} f^p(S_i).
    \]
    All of the elements within this sum are contained within $\fr(S_1, \ldots, S_a) = (f^p(S_1), \ldots, f^p(S_a))$. Thus we can deduce the value of $f(\alpha(x_1), \ldots, \alpha(x_n))$ uniquely from $\fr(x_1, \ldots, x_n)$, without reference to $f$, as required.
\end{proof}

\subsection{Formal system}

We will now describe a system of formal proofs that will help us reason about
additivity and dependency. First we define the semantics of this system.

\begin{definition}
    Suppose $S \subseteq A^n$ is a set of $n$-ary tuples, and $t \in A^n$ is an
    $n$-tuple. We write $S \vDash_{\A, \B} t$ (omitting $\A, \B$ if there is no chance for confusion) if there exists a function $h : B^S \to B$ (where we interpret $B^S$ as being a tuple indexed by $S$), such that, for any polymorphism $f \in \Pol^{(n)}(\A, \B)$, we have that
    \[
    f(t) = h(f(u) \mid u \in S).
    \]
    In other words, we write $S \vDash t$ if the value of $f(t)$ is uniquely determined by $f(u)$ for $u \in S$, in a manner independent of $f$.
\end{definition}

\begin{definition}
For any $\A$, we write $\trip_\A = \{ (r,s, s) \mid s, r \in A\} \cup \{ (s, r,
  s) \mid s, r \in A \}$ for the set of triples from $A^3$ with the last two
  elements equal or with the first and last elements equal. Write
  $\Delta_\A^n = \{ (x_1, \ldots, x_n ) \mid x_1, \ldots, x_n \in A, | \{ x_1, \ldots,
  x_n \}| \leq 2\}$ i.e.~$\Delta_\A^n$ contains all the $n$-ary tuples that
  contain at most two distinct elements.
\end{definition}

\begin{lemma}\label{lem:additivitySufficiency}
    $(\A, \B)$ is additive if for every $p, q \in A$ we have
    \[
        \trip_\A
        \vDash_{\A, \B} 
        (p, p, q).
    \]
\end{lemma}
\begin{proof}
    For $p, q \in A$, suppose $h_{qp}$ is the function that witnesses
        $\trip_\A
        \vDash_{\A, \B} 
        (p, p, q)$. Define, for $X, Y \in B^{A^2}$,
    \[
        (X + Y)(q, p) = h_{qp}(t_{x y z} \mid (x, y, z) \in \trip_\A),
    \]
    where $t_{sss} = X(s, s)$, $t_{r s s} = X(s,r)$ and $t_{srs} = Y(s,r)$, for $s, r \in A$. (This defines $t_{xyz}$ for all $(x, y, z) \in \trip_\A$.)
    
    We claim that this shows the additivity of $(\A, \B)$. Consider any
    polymorphism $f \in \Pol^{(n)}(\A, \B)$ and disjoint $S, T \subseteq [n]$.
    Write $\pi : [n] \to [3]$ for the function that takes $S$ to 1, $T$ to 2, and $[n] \setminus(S \cup T)$ to 3. By the definition of $h_{qp}$ and that of taking minors, we have
    \[
    (f^p(S \cup T))(q, p) = f_{qp}(S \cup T) = f^\pi(p, p, q)
    =h_{qp}(
    f^\pi(x,y,z) \mid (x, y, z) \in \trip_\A).
    \]
    Now, let $t_{xyz} = f^\pi(x, y, z)$. Observe that $t_{sss} = f^\pi(s,s,s) = f^p(S)(s, s)$. Also, $t_{rss} = f^\pi(r, s, s) = f^p(S)(s,r)$, and $t_{srs} = f^\pi(s, r, s) = f^p(T)(s, r)$. Thus we deduce that
    \[
     (f^p(S \cup T))(q, p) =
    h_{qp}(t_{xyz} \mid (x, y, z) \in \trip_\A) 
    = (f^p(S) + f^p(T))(q, p).
    \]
    Thus,
    \[
    f^p(S \cup T) = f^p(S) + f^p(T),
    \]
    as required by \Cref{def:additivity}.
\end{proof}

\begin{lemma}\label{lem:dependentcondition}
    Assuming $A = [a]$, $(\A, \B)$ is dependent if we have
    \[
    \Delta_\A^a \vDash_{\A, \B} (1, \ldots, a).
    \]
\end{lemma}

\begin{proof}
    Suppose $h$ witnesses that $\Delta_\A^a \vDash_{\A, \B}(1, \ldots, a)$. Note
    that every tuple in $\Delta_\A^a$ can be seen as a function from $[a]$ to some $S \subseteq A$ for $0 < |S| \leq 2$. In other words $h$ takes a tuple of elements from $B$, indexed by functions $\alpha : A \to S, S \in A^{\leq 2}$. Equivalently, we see that for any $a$-ary polymorphism $f$,
    \[
        f(1, \ldots, a) = h(f(\alpha(1), \ldots, \alpha(a)) \mid \alpha : A \to S, S \in A^{\leq 2}).
    \]
    The type of this function (not coincidentally) is exactly the same as the function that ought to witness the dependency of $(\A, \B)$; and indeed, we claim that it does in fact witness this.

    In other words, we must show that, for any $f \in \Pol^{(n)}(\A, \B)$ and $x_1, \ldots, x_n \in A$ we have that
    \[
    f(x_1, \ldots, x_n) = h(f(\alpha(x_1), \ldots, \alpha(x_n)) \mid A \to S, S \in A^{\leq 2}).
    \]
    Observe that for any $x_1, \ldots, x_n$, there exists a function $\pi : [n] \to [a]$ such that $f(x_1, \ldots, x_n) = f^\pi(1, \ldots, a)$; namely $\pi(i) = x_i$. Furthermore $f^\pi$ is an $a$-ary polymorphism; thus
    \begin{multline*}
        f(x_1, \ldots, x_n)
        =
        f^\pi(1, \ldots, a)
        =
        h( f^\pi( \alpha(1), \ldots, \alpha(a) ) \mid \alpha : A \to S, S \in A^{\leq 2})\\
        =h(f(\alpha(x_1), \ldots, \alpha(x_n)) \mid A \to S, S \in A^{\leq 2}),
    \end{multline*}
    as required.
\end{proof}

We now move on to a syntactic description of the formal proof system.

\begin{definition}
    Fix some $n \in \mathbb{N}$. Fix also some relation $R^\A$ of $\A$. For sets of $n$-ary tuples $S \subseteq A^n$ and $t \in A^n$, we define $S \vdash_{\A,R^\A} t$ as the minimal relation that satisfies the following.
    \begin{enumerate}
        \item If $S \subseteq T$ and $S \vdash t$ then $T \vdash t$.
        \item $t \vdash t$.
        \item If $S \vdash t$ and $t, T \vdash t'$ then $S, T \vdash t'$.
        \item If there exists a matrix with $n$ columns and $r$ rows, whose rows are $t_1, \ldots, t_r$, and whose columns are tuples of $R^\A$, then $t_2, \ldots, t_r \vdash t_1$.
    \end{enumerate}
    We omit $\A, R^\A$ if they are obvious from context.
\end{definition}
\begin{remark}
    Suppose $R^\A$ has arity $r$.
    Any judgement of the form $S \vdash_{\A, R^\A} t$ must have a finite proof using the rules above by minimality. 
    From this proof, we can create a \emph{proof-tree} where the vertices are $n$-ary tuples of $A$, the root is $t$, the leaves belong to $S$, and every non-leaf has $r-1$ children such that if $t_1$ is the non-leaf and $t_2, \ldots, t_r$ are its children, then the matrix whose rows are $t_1, \ldots, t_r$ has as its columns only tuples of $R^\A$. To create this tree, proceed inductively on the proof, from the conclusion backwards. The first and second rule do not modify the proof tree. The third rule corresponds to recursively constructing a subtree. The final rule is the only one that adds new vertices to the proof tree.
\end{remark}

\begin{lemma}\label{lem:syntaxSemantics}
    If $S \vdash_{\A,R^\A} t$ then $S \vDash_{\A, \B} t$ for any symmetric $\A$ and functional $\B$.
\end{lemma}
\begin{proof}
By minimality, it is sufficient to show that $\vDash_{\A, \B}$ satisfies all the rules satisfied by $\vdash_{\A, R^\A}$, which we do now rule-by-rule. For the following assume always that all tuples are of arity $n$.
\begin{enumerate}
    \item Suppose $S\subseteq T$ and $S \vDash t$. Thus there exists a function $h : A^S \to A$ such that, for any polymorphism $f \in \Pol^{(n)}(\A, \B)$, we have
    \[
        f(t) = h(f(u) \mid u \in S).
    \]
    Now, we define the function $h' : A^T \to A$ as follows:
    \[
    h'(x_u \mid u \in T) = h(x_u \mid u \in S).
    \]
    In other words, $h'$ ignores all inputs in $T \setminus S$ and otherwise acts like $h$. We claim that $h'$ witnesses that $T \vDash t$: for any polymorphism $f \in \Pol^{(n)}(\A, \B)$, we have
    \[
    f(t) = h(f(u) \mid u \in S) = h'(f(u) \mid u \in T).
    \]
    \item For any tuple $t$, we have $t \vDash t$; indeed, the identity function $\id_B : B \to B$ witnesses this fact.
    \item Suppose $S \vDash t$ and $t, T \vDash u$, as witnessed by $h : B^S \to B$ and $h' : B^{t + T} \to B$. If $t \in T$ then $T \vDash u$ and we have by the first rule that $S, T \vDash u$ as required; thus suppose $t \not \in T$. We will thus interpret $h'$ as having signature $h' : B \to B^T \to B$. With this interpretation, by definition, for any polymorphism $ f\in \Pol^{(n)}(\A, \B)$, we have
    \begin{align*}
        f(t) &=  h(f(v) \mid v \in S) \\
        f(u) &=  h'(f(t))(f(v) \mid v \in T).
    \end{align*}
    Now define $h'' : B^{S \cup T} \to B$ in the following way:
    \[
        h''(x_v \mid v \in S \cup T)
        =
        h'( h( x_v \mid v \in S) ) ( x_v \mid v \in T).
    \]
    We find that, for any $f \in \Pol^{(n)}(\A, \B)$, 
    \begin{multline*}
        f(u) = h'(f(t))(f(v) \mid v \in T)
        = h'( h( f(v) \mid v \in S)) ( f(v) \mid v \in T) \\
        = h''( f(v) \mid v \in S \cup T).
    \end{multline*}
    Thus we find that $S, T \vDash u$, as required.
    \item Suppose $R^\A$ has arity $r$, and suppose there exists a matrix with $n$ columns and $r$ rows, whose rows are $t_1, \ldots, t_r$, and whose columns are tuples of $R^\A$. Let $h : B^{\{t_2, \ldots, t_r\}} \to B$ be defined as follows, interpreting an element of $B^{\{t_2, \ldots, t_r\}}$ as an $(r-1)$-ary tuple in the natural way:
    \[
    h(x_2, \ldots, x_r) = R^\B(x_2, \ldots, x_r).
    \]
    Now, we claim that $h$ witnesses that $t_2, \ldots, t_r \vDash t_1$. Indeed, for any $n$-ary polymorphism $f \in \Pol^{(n)}(\A, \B)$, we have that $(f(t_1), \ldots, f(t_r)) \in R^\B$, and thus
    \[
    f(t_1) = R^\B(f(t_2),\ldots,f(t_r)) = h(f(t_2), \ldots, f(t_r)) = h( f(u) \mid u \in \{ t_2, \ldots, t_r\}).
    \]
\end{enumerate}
Thus we conclude that $\vDash_{\A, \B}$ satisfies the rules of $\vdash_{\A, R^\A}$, which implies our conclusion.
\end{proof}

For a relation $R$, we define its support $\supp(R)$ as  the elements that appear in at least one tuple of $R$. Formally,
\[
\supp(R) = \bigcup_{i = 1}^{\ar(R)} \{ a_i \mid (a_1, \ldots, a_{\ar(R)}) \in R \}.
\]

\begin{lemma}\label{lem:garbageValue}
    Suppose that $S \vdash_{\A, R^\A} t$ and $x \in \supp(R^\A)$. Then $S \times \supp(R^\A) \vdash_{\A, R^\A} (t, x)$.\footnote{Here, if $t = (t_1, \ldots, t_r)$ then $(t, x) = (t_1, \ldots, t_r, x)$, and likewise $S\times\supp(R^\A) = \{ (s_1, \ldots, s_r, x) \mid (s_1, \ldots, s_r) \in S, x \in \supp(R^\A)\}$.}
\end{lemma}
\begin{proof}
    We show this by induction on the formal proof that proves that $S \vdash_{\A, R^\A} t$. Based on the last step in the proof used to prove this fact, we have the following cases.
    \begin{enumerate}
        \item Suppose $S' \subseteq S$ and $S' \vdash t$. By the inductive hypothesis, we know that that $S' \times \supp(R^\A) \vdash (t, x)$. As $S' \times \supp(R^\A) \subseteq S \times \supp(R^\A)$, then $S \times \supp(R^\A) \vdash (t,x )$ as required.
        \item Suppose $S = \{ t \}$. In this case the result is immediate, as $S \times \supp(R^\A) \supseteq \{ (t, u)\} \vdash (t, u)$.
        \item Suppose $S = X \cup Y$, $X \vdash s$ and $s, Y \vdash t$. By the inductive hypothesis, for $y \in \supp(R^\A)$, $X \times \supp(R^\A) \vdash (s, y)$ and $\{ (s, y) \mid y \in \supp(R^\A) \}, Y \times \supp(R^\A) \vdash (t, x)$. We can therefore deduce, by this rule applied $|\supp(R^\A)|$ times, that $S \times \supp(R^\A) = (X \cup Y) \times \supp(R^\A)\vdash (t, x)$.
        \item Suppose that $S = \{ t_2, \ldots, t_r\}$ and there exists a matrix whose rows are $t, t_2, \ldots, t_r$ and all of whose columns are elements of $R^\A$. Suppose $(x,x_2, \ldots, x_r) \in R^\A$. Then by applying this rule we find that $(t_2, x), \ldots, (t_r, x_r) \vdash (t, x)$, which implies that $S \times \supp(R^\A) \vdash (t, x)$.
    \end{enumerate}
    Thus our conclusion follows.
\end{proof}

Any tuple $x \in A^n$ can be seen as a function from $[n]$ to $A$. Thus for any function $\pi : [m] \to [n]$, we let $x \circ \pi \in A^m$ be the tuple which, at position $i$, has value $x_{\pi(i)}$. For a set of tuples $S \in A^n$, we let $S \circ \pi = \{ x \circ \pi \mid x \in S \}$.

\begin{lemma}\label{lem:inversePoly}
    Suppose $S \subseteq A^n, t \in A^n$ and $\pi : [m] \to [n]$. If $S \vdash_{\A, R^\A} t$ then $S \circ \pi \vdash_{\A, R^\A} t \circ \pi$.
\end{lemma}
\begin{proof}
    First, consider any matrix $M$ whose rows are $r_1, \ldots, r_n$ and whose columns are $c_1, \ldots, c_m$. If we create a matrix $M'$ whose rows are $r_1 \circ \sigma, \ldots, r_n \circ \sigma$ for some $\sigma : [k] \to [m]$, we observe that the columns of this new matrix are $c_{\sigma(1)}, \ldots, c_{\sigma(k)}$. Thus, the columns of $M'$ are a subset of the columns of $M$.

    Consider the proof tree that proves that $S \vdash_{\A, R^\A} t$. If a vertex is labeled by a tuple $u$, transform it into a vertex labeled by $u \circ \pi$. By the previous observation, this remains a valid proof tree, whose leaves belong to $S \circ \pi$ and whose root is $t \circ \pi$. Thus $S \circ \pi \vdash t \circ \pi$.
\end{proof}

\subsection{Super-connectedness}

Rather than dealing directly with additivity and dependency, we will instead deal with a notion that implies them.

\begin{definition}
    We say that $\A$ is \emph{super-connected} if $\A$ has a relation $R^\A$ such that, for every $x, y, z \in A$ (perhaps even with $|\{x, y, z\}| < 3$), we have
    \[
    \trip_\A \vdash_{\A, R^\A} (x, y, z).
    \]
\end{definition}

Super-connectedness is a very useful concept: it implies both additivity and dependency.

\begin{lemma}\label{lem:superConnThenAdd}
    If $\A$ is super-connected then $(\A, \B)$ is additive for functional $\B$ where $\A \to \B$.
\end{lemma}
\begin{proof}
    An immediate consequence of \Cref{lem:additivitySufficiency} and \Cref{lem:syntaxSemantics}.
\end{proof}

\begin{lemma}\label{lem:superConnThenDepend}
    If $\A$ is super-connected then $(\A, \B)$ is dependent for functional $\B$ where $\A \to \B$.
\end{lemma}
We first prove two simple propositions. Assume for them that $\A$ is super-connected, witnessed by relation $R^\A$, and that $|A| > 1$.
\begin{proposition}\label{prop:suppRA}
    $\supp(R^\A) = A$.
\end{proposition}
\begin{proof}
    Suppose for contradiction that this is not the case. Let $x \in A \setminus \supp(R^\A)$, and take $y \in A$ such that $x \neq y$. Now,
    \[
    \trip_\A \vdash_{\A, R^\A} (x, x, y).
    \]
    Consider any proof tree concluding in this judgement. The root vertex cannot
    have any children, since no matrix containing $(x, x, y)$ can have its
    columns belong to $R^\A$, as $x \not \in \supp R^\A$. Thus $(x, x,y)$ is a
    leaf vertex in the tree, and $(x, x, y) \in \trip_\A$. This is not possible, as $x \neq y$.
\end{proof}

\begin{proposition}\label{prop:An}
    For any $x_1, \ldots, x_n \in A$, we have that $\Delta_\A^{n} \vdash (x_1, \ldots, x_n)$.
\end{proposition}
\begin{proof}
    We prove this fact by induction. We do induction over the lexicographical ordering on $\{ (n, d) \in \mathbb{N}^2 \mid n \geq d \}$.
    For any $(n, d) \in \mathbb{N}^2$ with $n \geq d$ we prove that $\Delta_\A^n \vdash (x_1, \ldots, x_n)$ whenever $| \{ x_1, \ldots, x_n \} | \leq d$.
    As our base case, note that the result is immediate when $n \leq 2$, as
    $(x_1, \ldots, x_n) \in \Delta_\A^n$ in this case. Thus suppose $n \geq 3$.
    
    First suppose that $x_1, \ldots, x_n$ contains a duplicated pair of values, say $x_{n-1} = x_n$. Now, define $\pi : [n-1] \to [n]$, where $\pi(i) = i$. By the inductive hypothesis, we know that $\Delta_\A^{n-1} \vdash (x_1, \ldots, x_{n-1})$. Thus, by \Cref{lem:inversePoly}, $\Delta_\A^{n-1} \circ \pi \vdash (x_1, \ldots, x_{n-1}) \circ \pi = (x_1, \ldots, x_{n-1}, x_{n-1}) = (x_1, \ldots, x_n)$. Furthermore, every tuple of $\Delta_\A^{n-1} \circ \pi$ contains at most 2 values, so it belongs to $\Delta_\A^{n}$. Thus in this case $\Delta_\A^{n} \vdash (x_1, \ldots, x_n)$.

    Now suppose that $x_1, \ldots, x_n$ are all distinct. Since $\A$ is
    super-connected, we find that $\trip_\A \vdash (x_1, x_2, x_3)$. Thus by
    \Cref{lem:garbageValue}, we find that $\trip_\A \times {\supp(R^\A)}^{n-3} \vdash
    (x_1, \ldots, x_n)$. By \Cref{prop:suppRA}, $\supp(R^\A) = A$, so $\trip_\A \times A^{n-3} \vdash (x_1, \ldots, x_n)$. But, every tuple $t \in \trip_\A \times A^{n-3}$ contains at most $n-1$ distinct values, so we can apply the inductive hypothesis to them, and find that $\Delta_\A^n \vdash t$. Thus $\Delta_\A^n \vdash (x_1, \ldots, x_n)$, as required. This completes the proof.
\end{proof}

\begin{proof}[Proof of \Cref{lem:superConnThenDepend}]
    $\Delta_\A^a \vdash (1, \ldots, a)$ by \Cref{prop:An}. Thus, by
    \Cref{lem:syntaxSemantics} we have that $\Delta_\A^a \vDash_{\A, \B} (1, \ldots, a)$, and therefore dependency follows by \Cref{lem:dependentcondition}.
\end{proof}

\begin{lemma}\label{lem:diamSuperConn}
Suppose $\A$ has a symmetric relation $R^\A$ of arity at least 3 with $\diam(A, R^\A) \leq 1$.
Then $\A$ is super-connected.
\end{lemma}
\begin{proof}
Take $x, y, z \in A$; we must show that $\trip_\A \vdash (x, y, z)$.
We have several cases.
\begin{description}
\item[$\bm{x = z}$ or $\bm{y = z}$.] In this case $(x, y, z) \in \trip_\A$, so there is nothing left to be proved.
\item[$\bm{x = y \neq z}$.] In this case $\dist_{R^\A}(x, z) = 1$, so there exists an edge $(x, z, a_3, \ldots, a_r)\in R^\A$. Consider the matrices
\[
M = 
    \begin{pmatrix}
    x & x & z \\
    z & a_3 & x \\
    a_3 & z & a_3 \\
    a_4 & a_4 & a_4 \\
    \vdots & \vdots & \vdots \\
    a_r & a_r & a_r \\
    \end{pmatrix},
\qquad
    N = 
    \begin{pmatrix}
    z & a_3 & x \\
    x & z & z \\
    a_3 & x & a_3 \\
    a_4 & a_4 & a_4 \\
    \vdots & \vdots & \vdots \\
    a_r & a_r & a_r \\
    \end{pmatrix}.
\]
Thus we deduce that $(z, a_3, x), \trip_\A \vdash (x, x, z)$ and $\trip_\A \vdash (z, a_3, x)$, and thus $\trip_\A \vdash (x, x, z) = (x, y, z)$.
\item[$\bm{x \neq y, y \neq z, x \neq z}$.]
In this case $\dist_{R^\A}(x, z) = \dist_{R^\A}(x, y) = 1$, so there must exist edges $(x, z, a_3, \ldots, a_r), (x, y, b_3, \ldots, b_r) \in R^\A$. 
Consider the following matrix
    \[
    \begin{pmatrix}
    x & y & z \\
    z & x & x \\
    a_3 & b_3 & a_3 \\
    \vdots & \vdots & \vdots \\
    a_r & b_r & a_r \\
    \end{pmatrix}.
    \]
Thus we find that $(z, x, x), (a_3, b_3, a_3), \ldots, (a_r, b_r, a_r) \vdash (x, y, z)$, whence $\trip_\A \vdash (x, y, z)$.
\end{description}
Thus our conclusion follows in all cases.
\end{proof}

\begin{lemma}\label{lem:3superConn}
    If $\A$ has a connected symmetric relation $R^\A$ of arity 3 then $\A$ is super-connected.
\end{lemma}
\begin{proof}
    We see $(A, R^\A)$ as a connected 3-uniform hypergraph. For $x, y \in A$ recall that $\dist(x, y)$ is the distance between $x$ and $y$ in this hypergraph. We show that $\trip_\A \vdash_{\A, R^\A} (x, y, z)$ for all $x, y, z \in A$ by lexicographic induction on $\minmax(\dist(x, z), \dist(y, z))$ (i.e.~we first order by the minimum, and in case of equality the maximum). In all the cases that follow assume $\dist(x, z) \leq \dist(y, z)$.
    \begin{description}
    \item[$\bm{\dist(x, z) = 0}$.] In this case $(x, y, z) = (x, y, x) \in \trip_\A$ so there is nothing left to prove.
    \item[$\bm{\dist(x, z) = \dist(y, z) = 1}$.] In this case, there exists
      edges $(x, a, z)$, $(y,b , z)$. Now consider the following matrices:
    \[
    M=
    \begin{pmatrix}
        x & y & z \\
        z & z & x \\
        a & b & a
    \end{pmatrix},
    \qquad
    N=
    \begin{pmatrix}
        z & z & x \\
        a & x & z \\
        x & a & a
    \end{pmatrix},
    \qquad
    P =
    \begin{pmatrix}
    a & x & z \\
    x & z & x \\
    z & a & a
    \end{pmatrix}.
    \]
    Thus we conclude that $(a, b, a), (z, z, x) \vdash (x, y, z)$, $(a, x, z), (x, a, a) \vdash (z, z, x)$ and also $(x, z,x), (z, a,a) \vdash (a, x, z)$, from where the conclusion follows.
    \item[$\bm{\dist(y, z) > \dist(x, z) = 1}$.] In this case, there exist edges $(x, a, z)$ and $(y, b, y')$ such that $\dist(y', z) = \dist(y, z) - 1$. Furthermore, as $\dist(y', z) > 0$, there exists an edge $(y', c, d)$ such that
    \[
    \dist(c, z), \dist(d, z) \leq \dist(y', z) \leq \dist(y, z) - 1,
    \]
    Now consider matrices
    \[
    M = 
    \begin{pmatrix}
        x & y & z \\
        a & b & a \\
        z & y' & x
    \end{pmatrix},
    \qquad
    N = \begin{pmatrix}
        z & y' & x \\
        a & c & a \\
        x & d & z
    \end{pmatrix}.
    \]
    Hence $(z, y', x), \trip_\A \vdash (x, y, z)$ and $(x, d, z), \trip_\A \vdash (z, y', x)$. Since $\dist(d, z) < \dist(y, z)$ we can apply the inductive hypothesis to $(x, d, z)$, so $\trip_\A \vdash (x, y, z)$ as required.
    
    \item[$\bm{\dist(y, z) \geq \dist(x, z) > 1}$.] In this case, there exist edges $(x, a, x'), (z, b, z')$ such that it is true that $\dist(a, z'), \dist(b, x') < \dist(x, z)$ --- take edges from $x, z$ towards the other one. (Such edges only exist when $\dist(x, z) > 1$ as assumed.) Furthermore, by connectedness, an edge $(y, c,d)$ exists.
    Thus consider the matrix
    \[
    \begin{pmatrix}
        x & y & z \\
        a & c & z' \\
        x' & d & b
    \end{pmatrix}.
    \]
    Hence $(a, c, z'), (x,d, b) \vdash (x, y,z )$. By induction $\trip_\A \vdash (a, c, z')$ and $\trip_\A \vdash (x', d, b)$, so we have $\trip_\A \vdash (x, y, z)$.
    \end{description}
    Thus we find by induction that $\A$ is super-connected.
\end{proof}

\begin{lemma}\label{lem:4superConn}
    If $\A$ has a connected symmetric relation $R^\A$ of arity 4 then $\A$ is super-connected.
\end{lemma}
\begin{proof}
    We see $(A, R^\A)$ as a connected 4-uniform hypergraph. We show that $\trip_\A \vdash_{\A, R^\A} (x, y, z)$ for all $x, y, z \in A$ by lexicographic induction on $\minmax(\dist(x, z), \dist(y, z))$. In all the cases that follow assume $\dist(x, z) \leq \dist(y, z)$.
    \begin{description}
    \item[$\bm{\dist(x, z) = 0}$.] In this case $(x, y, z) = (x, y, x) \in \trip_\A$ so there is nothing left to prove.
    \item[$\bm{\dist(x, z) = \dist(y, z) = 1}$.] In this case, there exists
      edges $(x, a, a', z)$, $(y, b, b', z)$. Now consider the following matrices:
    \[
    M=
    \begin{pmatrix}
        x & y & z \\
        z & z & x \\
        a & b & a \\
        a' & b' & a'
    \end{pmatrix},
    \qquad
    N=
    \begin{pmatrix}
        z & z & x \\
        a & x & z \\
        x & a & a \\
        a' & a' & a' \\
    \end{pmatrix},
    \qquad
    P =
    \begin{pmatrix}
    a & x & z \\
    x & z & x \\
    z & a & a \\
    a' & a' & a'
    \end{pmatrix}.
    \]
    Hence, $\trip_\A \vdash (x,y, z)$.
    \item[$\bm{\dist(y, z) > \dist(x, z) = 1}$.] In this case, there exist edges $(x, a, a', z)$ and $(y, b, b', y')$ such that $\dist(y', z) = \dist(y, z) - 1$. Furthermore, as $\dist(y', z) > 0$, there exists an edge $(y', c, d, e)$ such that
    \[
    \dist(c, z), \dist(d, z), \dist(e, z) \leq \dist(y', z) \leq \dist(y, z) - 1,
    \]
    Now consider matrices
    \[
    M = 
    \begin{pmatrix}
        x & y & z \\
        a & b & a \\
        a'& b'& a' \\
        z & y' & x
    \end{pmatrix},
    \qquad
    N = \begin{pmatrix}
        z & y' & x \\
        a & c & a \\
        a'& d & a' \\
        x & e & z
    \end{pmatrix}.
    \]
    Hence $(z, y', x), \trip_\A \vdash (x, y, z)$ and $(x, d, z), \trip_\A (z, y', x)$. Since $\dist(d, z) < \dist(y, z)$ we can apply the inductive hypothesis to $(x, d, z)$, so $\trip_\A \vdash (x, y, z)$ as required.
    
    \item[$\bm{\dist(y, z) \geq \dist(x, z) > 1}$.] In this case, there exist edges $(x, a, a', x'), (z, b, b', z')$ such that it is true that $\dist(a, z'), \dist(a', z'), \dist(b, x'), \dist(b', x') < \dist(x, z)$ --- take edges from $x, z$ towards the other one. Furthermore, as $\dist(y, z) > 2$ there exist edges $(y, c, c', y')$ and $(y', d, d', y'')$, with $\dist(y'', z) = \dist(y, z) - 2$.
    Thus consider the matrices
    \[
    M = \begin{pmatrix}
        x & y & z \\
        a' & c & z' \\
        x' & c' & b' \\
        a & y' & b \\
    \end{pmatrix},
    \qquad
    N = \begin{pmatrix}
        a & y' & b \\
        a' & d & z' \\
        x' & d' & b' \\
        x & y'' & z
    \end{pmatrix}.
    \]
    Note that $\dist(a', z'), \dist(x', b') < \dist(x, z)$, and furthermore $\dist(y'', z) < \dist(y, z)$, so by induction $\Gamma_\A \vdash (x, y, z)$.
    \end{description}
    Thus we find by induction that $\A$ is super-connected.
\end{proof}

With all of this, we are finally ready to prove \Cref{thm:diamAdditivityDependency} and \Cref{thm:4AdditivityDependency}.
\begin{proof}[Proof of \Cref{thm:diamAdditivityDependency} and \Cref{thm:4AdditivityDependency}]
The structures from \Cref{thm:diamAdditivityDependency} are super-connected by \Cref{lem:diamSuperConn}. The structures from \Cref{thm:4AdditivityDependency} are super-connected by \Cref{lem:3superConn} and \Cref{lem:4superConn}. This is sufficient for additivity by \Cref{lem:superConnThenAdd} and for dependency by \Cref{lem:superConnThenDepend}.
\end{proof}

\section{Dichotomy}

In this section we will prove our main result.

\thmdichotomy*

Before proving \Cref{thm:dichotomy}, we will prove three interesting corollaries. 

\corBoolDich*

For the proof of \Cref{cor:BooleanDichotomy}, we will need a simple lemma.

\begin{lemma}\label{lem:specialCase}
Suppose $\A$ is Boolean and symmetric, $\B$ is functional, and every relation of $\A$ is either binary or contains only constant tuples. Then, $\PCSP(\A, \B)$ is solvable in polynomial time by $\AIP$ and is finitely tractable.
\end{lemma}
\begin{proof}
Consider any $h : \A \to \B$. 
Suppose $h(0) = h(1)$. Then every relation in $\B$ contains a constant tuple of the form $(h(0), \ldots, h(0))$; in this case, $\PCSP(\A, \B)$ is trivially solved by $\AIP$ and is finitely tractable. Thus suppose $h(0) \neq h(1)$. Any empty relation in $\A$ can be removed (together with the corresponding relation in $\B$) as it does not affect  $\Pol(\A, \B)$ and the complexity of $\PCSP(\A,\B)$.

Since $\B$ is functional, the binary relations of $\A$ that do not contain only constant tuples must be the binary disequality. To see why, consider any relation $R^\A$ in $\A$ that contains the tuple $(0, 1)$. $R^\A$ cannot contain $(0, 0)$ or $(1, 1)$, since the corresponding relation $R^\B$ in $\B$ contains $(h(0), h(1))$ and if it contained $(h(0), h(0))$ or $(h(1), h(1))$ it would not be functional. Since $\A$ is symmetric, $R^\A$ also contains the tuple $(1,0)$. Thus $R^\A = \{(0, 1), (1, 0)\}$ is the disequality relation. It follows that every relation in $\A$ is either a binary disequality, or consists only of constant tuples. In this case, $\CSP(\A)$ is solved by $\AIP$
and thus $\PCSP(\A, \B)$ is solved by $\AIP$ and is finitely tractable.
\end{proof}

\begin{proof}[Proof of \Cref{cor:BooleanDichotomy}]
If $\A$ contains only binary relations or relations that contain only constant
  tuples, the conclusion follows by \Cref{lem:specialCase}. Otherwise,
$\A$ has a relation $R^\A$ of arity at least 3 with a non-constant tuple i.e.~for which $\diam(A, R^\A) = 1$, and the conclusion follows from 
\Cref{thm:dichotomy} together with super-connectivity of such structures (cf.~\Cref{lem:diamSuperConn}):
super-connectedness implies additivity (cf.~\Cref{lem:superConnThenAdd})
and dependency (cf.~\Cref{lem:superConnThenDepend}).
\end{proof}

\corTernDich*

\begin{proof}
First, observe that the conclusion holds trivially if $R^\B$ has arity 1 as in
  this case $\CSP(\B)$ is trivially tractable. Second, suppose $R^\B$ has arity 2. In this case we claim that $\CSP(\B)$ is solvable in polynomial time. Indeed, an easy modification of the algorithm for 2-colouring will solve this problem. Thus $\PCSP(\A, \B)$ is tractable in this case.

  Now, suppose $R^\B$ has arity 3 or 4. The conclusion holds for connected $\A$ by the super-connectedness of structures $\A$ that have a relation of arity 3 or 4 that is connected. 
  Thus suppose that $\A$ is not connected, and that its connected components are $\A_1, \ldots, \A_k$ i.e.~$\A = \A_1 + \cdots + \A_k$. If $\PCSP(\A_i, \B)$ is \NP-hard for some $i \in [k]$ then $\PCSP(\A, \B)$ will also be \NP-hard (since there is a trivial reduction from $\PCSP(\A_i, \B)$ to $\PCSP(\A, \B)$, as $\A_i \to \A$). Thus suppose $\PCSP(\A_i, \B)$ is solvable in polynomial time for all $i \in [k]$. Consider any input hypergraph $\X = \X_1 + \cdots + \X_n$, where $\X_1, \ldots, \X_n$ are connected. Since the homomorphic image of a connected hypergraph is connected, if $\X_i \to \A$ then $\X_i \to \A_j$ for some $j \in [k]$. Thus, for the decision version of $\PCSP(\A, \B)$, it is sufficient to see if, for all $i \in [n]$ there exists some $j \in [k]$ such that $\X_i$ is a \textsc{Yes}-instance of $\PCSP(\A_j, \B)$. If so then $\X$ is a \textsc{Yes}-instance overall. For the search version something similar happens: see if we can produce a homomorphism $\X_i \to \B$ for each $\X_i$ by running the algorithm for $\PCSP(\A_j, \B)$ for each $\A_j$, and combine these homomorphisms to find a homomorphism $\A \to \B$.
\end{proof}

\begin{remark}\label{rem:notbyAIP}
Unlike in all other results in this paper, solvability by $\AIP$ is not proved
in \Cref{cor:TernDich}, only polynomial-time solvability.
As far as we know, it well may be that the PCSPs considered in \Cref{cor:TernDich} 
are not solved by $\AIP$ (or even $\BLPAIP$).
The way that tractability is deduced in the proof of
    \Cref{cor:TernDich} is as follows: If $\A_1, \ldots, \A_k$ are connected and $\PCSP(\A_i, \B)$ is tractable, then $\PCSP(\A_1 + \cdots + \A_k, \B)$ is tractable. We give a concrete example of a template with a single symmetric relation of arity 6 that shows that this reduction \emph{does not} preserve solvability by $\AIP$, or even by $\BLPAIP$.
  Note however that \Cref{cor:TernDich} only applies to
  relations of arity at most 4, so the corollary might still be strengthened to show solvability by $\AIP$.

  Let $\mathbb{Z}' = \{ x' \mid x \in \mathbb{Z} \}$ be a disjoint copy of $\mathbb{Z}$. Then, let $\B=\A_1+\A_2$, where
    \begin{align*}
    \A_1 & = ( \{ 0, 1 \} ; \{ (x_1, \ldots, x_6) \mid x_1 + \cdots + x_6 \equiv 1 \pmod 2 \} ),\\
    \A_2 & = ( \{ 0', 1', 2' \} ; \{ (x_1, \ldots, x_6) \mid x_1 + \cdots + x_6 \equiv 2' \pmod{3'} \} ).
    \end{align*}
 Note that $\A_1, \A_2, \B$ are all functional and symmetric. Furthermore,
  $\PCSP(\A_1, \B)$ and $\PCSP(\A_2, \B)$ are solved by $\AIP$ by
  \Cref{thm:aip}: a $(2k+1)$-ary alternating polymorphism is
    \[
    (x_1, y_1, \ldots, x_k, y_k, x_{k+1}) \mapsto x_1 - y_1 + \cdots + x_k - y_k + x_{k+1} \pmod{m},
    \]
    where $m = 2$ for $\A_1$ and $m = 3'$ for $\A_2$. We will now show that
    $\BLPAIP$ fails to solve $\PCSP(\B, \B) = \CSP(\B)$; to see why, suppose
    that it does solve it. Thus, by \Cref{thm:blpaip}, $\Pol(\B,\B)$
    should contain a 2-block symmetric polymorphism of arity 5, say $f \in
    \Pol^{(5)}(\B,\B)$. For the ease of notation, suppose that the first block of symmetry of $f$ contains the first 3 inputs, and the second block of symmetry contains the last 2 inputs (rather than the blocks being based on parity). Now consider the following matrix
    \[
    \left(
    \begin{array}{ccc|cc}
        1' & 0' & 0' &
        1 & 0 \\
        0' & 1' & 0' &
        0 & 1\\
        0' & 0' & 1' &
        1 & 0\\
        1' & 0' & 0' &
        0 & 1\\
        0' & 1' & 0' &
        1 & 0\\
        0' & 0' & 1' &
        0 & 1\\
    \end{array}
    \right)
    \]
    Every column of this matrix is an element of $R^\B$; thus $f$ applied to every row gives a tuple of $R^\B$. But, due to block-symmetry, the image of every row through $f$ is the same! This contradicts the lack of constant tuples in $\B$.

The non-solvability of $\PCSP(\B,\B)$ by $\AIP$ is relevant in view of
  \Cref{conj:aip} --- it shows that the ``$\AIP$ being
  a universal algorithm for $\PCSP(\inn{1}{3},-)$'' part of \Cref{conj:aip} cannot be extended to \emph{arbitrary} symmetric templates. We believe that it might hold true for any \emph{connected} $\A$ with one symmetric relation.
\end{remark}

\corDiamDich*

\begin{proof}
  Such structures are super-connected by \Cref{lem:diamSuperConn}, which
  implies additivity and dependency of $(\A,\B)$ by \Cref{lem:superConnThenAdd} and
  \Cref{lem:superConnThenDepend}, respectively.
\end{proof}

We will now move on to a proof of \Cref{thm:dichotomy}.
Suppose $(\A,\B)$ is dependent and additive. Suppose generally that $A = [a]$.

\begin{definition}
  Consider a polymorphism $f\in \Pol^{(n)}(\A, \B)$. We call it
  \emph{$k$-degenerate} if there exist $x_1, \ldots, x_k \in \range(f^p)$ such
  that for any $S_1, \ldots, S_k \subseteq [n]$ for which $f^p(S_i) = x_i$ we
  have that not all $S_i$ are disjoint. Note that no polymorphism can be
  1-degenerate as a single set is a disjoint family.

For any polymorphism $f \in \Pol^{(n)}(\A, \B)$, we call a set $S \subseteq [n]$ a \emph{hard set} if, for any $T \supseteq S$, we have $f^p(T) \neq f^p(\emptyset)$.%
\footnote{These two notions are similar to those of \emph{unbounded antichains} and \emph{fixing sets} in~\cite{Ficak19:icalp}. The notion of hard-set is similar to the notion of an $f(\emptyset)$-avoiding set from~\cite{DRS05,Wrochna22}.}
\end{definition}

We will prove \Cref{thm:dichotomy} using the following two cases. For the following, define $N_d = \max(1 + |B|^{a^2} a^{2r_{\max}}, 3)$ and $N_h = |B|^{a^2}$, where $r_{\max}$ is the maximum arity of any relation in $\A$.

\begin{theorem}\label{thm:genTractability}
If $\Pol(\A, \B)$ contains a polymorphism that is not $k$-degenerate, for any $k \leq N_d$, and that has no hard sets of size at most $N_h$, then $\PCSP(\A, \B)$ is solved by $\AIP$ and is finitely tractable.
\end{theorem}

\begin{theorem}\label{thm:genHardness}
If every polymorphism within $\Pol(\A, \B)$ is $k$-degenerate for some $k \leq N_d$, or has a hard set of size at most $N_h$, then $\PCSP(\A, \B)$ is \NP-hard.
\end{theorem}

These two theorems will be proved in their own sections later.

\begin{proof}[Proof of \Cref{thm:dichotomy}]
A result of
\Cref{thm:genTractability} and \Cref{thm:genHardness}.
\end{proof}

\begin{remark}
    The following turns out to be an equivalent characterisation of the solvability of $\PCSP(\A,\B)$ for symmetric $\A$ and functional $\B$. Suppose $\A$ has domain $A$ and relations $R_1^\A, \ldots, R_k^\A$. For any positive integer $m$, define a structure $\A_m$ in the following way. The domain of $\A_m$ is the free $\mathbb{Z}_m$-module of functions $\mathbb{Z}_m^A$. We will write elements of this module as formal sums of the form $\sum_{a \in A} x_i \overline{a}$, where $x_i \in \mathbb{Z}_m$ and $\{ \overline{a} \mid a \in A \}$ is a basis for $\mathbb{Z}_m^A$.
    Extend the map $x \mapsto \overline{x}$ to tuples of functions, in the following way: $\overline{(x_1, \ldots, x_n)} = (\overline{x_1}, \ldots, \overline{x_n})$.
    To define the relation $R_k^{\A_m}$, consider the set of tuples of functions $S = \{ \overline{t} \mid t \in R_k^\A \}$. Tuples of functions also form a free $\mathbb{Z}_m$-module; thus take $R_k^{\A_m}$ to be the minimal affine space containing $S$. Equivalently, $R_k^{\A_m}$ is the set of tuples that are equivalent to some $t \in R_k^\A$ modulo $\{ t - t' \mid t, t' \in R_k^\A \}$, if the set of tuple of functions from $\mathbb{Z}_m^A$ is seen as merely an Abelian group. 
    
    With the relational structure $\A_m$ thus defined, and noting that $\A \to \A_m$ always via the homomorphism $x \mapsto \overline{x}$, we note that our result is equivalent to the following: for symmetric $\A$ and functional $\B$ where $(\A, \B)$ is additive and dependent, $\PCSP(\A, \B)$ is solvable in polynomial time if and only if $\A \to \A_m \to \B$ for some positive integer $m$. That this condition is sufficient is clear, as $\CSP(\A_m)$ is solved by \AIP{} as per \Cref{lem:sandwich}. Necessity follow from our proof of \Cref{thm:genTractability} and \Cref{thm:genHardness}.

    We note also that $\A_m$ can also be described as a free structure~\cite{BBKO21}. Namely, define $\mathcal{Z}_m$ to be a minion such that 
    \[
    \mathcal{Z}_m^{(n)} = \{ (x_n, \ldots, x_n) \mid 0 \leq x_1, \ldots, x_n < m, \sum_{i = 1}^n x_i \equiv 1 \pmod{m}\},
    \]
    and let minoring be defined as follows. For $x = (x_1, \ldots, x_n) \in \mathcal{Z}_m^{(n)}$, and for $\pi : [n] \to [n']$, define $y = x^\pi$ by $y_i = \sum_{\pi(j) = i} x_j \bmod m$. With this in mind, it can be seen that $\A_m = \mathbf{F}_{\mathcal{Z}_m}(\A)$.
\end{remark}

\subsection{\texorpdfstring{Proof of \Cref{thm:genTractability}}{Proof of Theorem 42}}

In this section we assume that $\Pol^{(n)}(\A, \B)$ has a polymorphism $f$ of arity $n$ that is not $k$-degenerate for $k$ at most $N_d$, and has no hard sets of size at most $N_h$. Given this, we will prove that $\PCSP(\A, \B)$ is solved by $\AIP$ and is finitely tractable.

\begin{definition}
Define $0 = f^p(\emptyset)$ and $1 = f^p([n])$.
\end{definition}

\begin{lemma}
$(\range(f^p), +, 0)$ forms a group.%
\footnote{This group happens to be Abelian, but this is not needed for the proof.}
\end{lemma}
\begin{proof}
We prove this in a few parts.
\begin{description}
\item[Closure, well-definedness.] Consider
  $x, y \in \range(f^p)$. As $f$ is not 2-degenerate, there exist disjoint $S, T$ such that $f^p(S) = x, f^p(T) = y$. Thus
  $x + y = f^p(S) + f^p(T) = f^p(S \cup T) \in \range(f^p)$,
  so $+$ is closed and well-defined.
\item[Associativity.] Consider any $x, y, z \in \range(f^p)$. Since $f$ is not 3-degenerate, there exist disjoint $S, T, U \subseteq [n]$ such that $f^p(S) = x, f^p(T) = y, f^p(U) = z$. Thus,
\[
x + (y + z) = x + f(T \cup U)
= f(S \cup (T \cup U))
= f((S \cup T )\cup U)
= f(S \cup T) + z
= (x + y) + z.
\]
\item[Identity element.] Consider any $x \in \range(f^p)$. Suppose $f^p(S) = x$ for some $S \subseteq [n]$. Thus,
$x + 0 = f^p(S \cup \emptyset) = f^p(S) = x$.
\item[Inverses.] Consider any $x \in \range(f^p)$. Suppose that $f(S) = x$; by \Cref{lem:genSmallSet}, some $T \subseteq S$ exists with size at most $|B|^{a^2}$ such that $f^p(T) = f^p(S) = x$. Since $f$ has no hard sets of size at most $N_h = |B|^{a^2}$, $T$ is not a hard set, and thus some $U \supseteq T$ exists such that $f^p(U) = 0$. Thus $x + f^p(U \setminus T) = f^p(T) + f^p(U \setminus T) = f^p(U) = 0$, so $x$ has an inverse.
\end{description}
Thus we conclude that $(\range(f^p), +, 0)$ is a group.
\end{proof}

\begin{definition}
Let $G$ be the Abelian subgroup of
$(\range(f^p), +, 0)$
generated by $1 = f^p([n])$. Let $m$ be the order of $1$ in $G$. Thus $G \cong \mathbb{Z}_m$. Note that $m \leq |\range(f^p)| \leq |B|^{a^2}$. We will identify $\mathbb{Z}_m$ with $G$ (e.g.~we allow ourselves to write $1 + 1 = 2$, provided $m \geq 3$, where $1, 2 \in \range(f^p)$).

Define the Abelian group $(H, +) = G^a$.
We will identify $H$ with $\mathbb{Z}_m^a$. We will also define $0$ to be the 0 element in $H$ as well as $G$.

For any $i \in [a]$, define $\overline{i} \in H$ as the unit vector that has a 1 at position $i$. For some tuple $(x_1, \ldots, x_r) \in {[a]}^r$, define $\overline{(x_1, \ldots, x_r)} = (\overline{x_1}, \ldots, \overline{x_r}) \in H^r$. Define $0$ to be the zero vector in $H^r$ as well.\footnote{We can see the elements of $H$ as frequency vectors modulo $m$. Indeed, for $x_1, \ldots, x_n \in [a]$, $\overline{x_1} + \cdots + \overline{x_n}$ counts the number of appearances of $1, 2, \ldots, a$ modulo $m$ among $x_1, \ldots, x_n$. In line with this, the elements of $H^r$ can be seen as tuples of $r$ frequency vectors. Under this view, for $t_1, \ldots, t_n \in {[a]}^r$, the sum $\overline{t_1} + \cdots + \overline{t_n}$ is a tuple of $r$ frequency vectors, where the $i$-th frequency vector counts the frequencies of the elements of $[a]$ among the $i$-th elements of the tuples $t_1, \ldots, t_n$, modulo $m$.}

For any relation $R^\A$ of $\A$ of arity $r$, define $M(R^\A)$ to be the subgroup of $H^r$ generated by $\overline{p} - \overline{q}$ for $p, q \in R^\A$. Since $H^r$ is Abelian, $M(R^\A)$ is a normal subgroup.
\end{definition}

\begin{lemma}\label{lem:matrix}
Fix some relation $R^\A$ of $\A$; suppose it has arity~$r$. Let $t$ be some tuple of $R^\A$. Define $M = M(R^\A)$.
Consider any $(a_1, \ldots, a_r) \in H^r$ such that $(a_1, \ldots, a_r) \equiv \overline{t} \bmod M$. There exists a matrix $(x_{ij})$ with $N \leq N_d$ columns and $r$ rows, where $N \equiv 1 \bmod m$, with elements in $[a]$, such that each column is a tuple of $R^\A$, and, for each row $i$, we have
\[
\sum_{j = 1}^N \overline{x_{ij}} = a_i.
\]
\end{lemma}
\begin{proof}
Note that every element in $H^r$ has order that divides $m$ (since $H^r \cong {(G^{a})}^r \cong {(\mathbb{Z}_m^a)}^r$). Thus, since $(a_1, \ldots, a_r) \equiv \overline{t} \bmod M$, and since $M$ is generated by $\overline{p} - \overline{q}$ for $p, q \in R^\A$, it follows that there exist coefficients $c_{pq} \in \{0, \ldots, m-1\}$ for $p, q \in R^\A$ such that
\begin{equation}\label{eqn1}
(a_1, \ldots, a_r) = \overline{t} + \sum_{p, q \in R^\A} c_{pq}(\overline{p} - \overline{q})
=
\overline{t} + \sum_{p, q \in R^\A} c_{pq}\overline{p} + (m - c_{pq}) \overline{q}.
\end{equation}
This indicates the matrix we will use: let $(x_{ij})$ be a matrix whose first column is $t$, and, for each $p, q \in R^\A$, has $c_{pq}$ columns equal to $p$ and $m - c_{pq}$ columns equal to $q$. Clearly we use $N = 1 + m |R^\A|^2 \leq 1 +|B|^{a^2} a^{2r_{\max}} \leq N_d$ columns, and $N \equiv 1 \bmod m$. To see why $\sum_{j =1}^N \overline{x_{ij}} = a_i$ for each $i$, note that this condition is equivalent to $(a_1, \ldots, a_r) = \sum_{j = 1}^N \overline{c_j}$, where $c_1, \ldots, c_N$ are the columns of the matrix. But this is precisely Equation~\eqref{eqn1}. Thus we have created the required matrix.
\end{proof}

\begin{lemma}\label{lem:polys}
    Suppose $f = h \circ \fr$.
    For every $N \leq N_d$ such that $N \equiv 1 \bmod m$, the function $h_N : A^N \to B$ defined by
    \[
    h_N(x_1, \ldots, x_N) = h\left( \sum_{i = 1}^N \overline{x_i}\right)
    \]
    is a polymorphism of $(\A, \B)$.
\end{lemma}
\begin{proof}
    By assumption, $f$ is not $N$-degenerate. Thus there exist disjoint subsets $S_1, \ldots, S_N$ of $[n]$ where $f^p(S_1) = \cdots = f^p(S_N) = f^p([n]) = 1$. Let $T = [n] \setminus (S_1 \cup \ldots \cup S_N)$. Note that $S_1, \ldots, S_N, T$ form a partition of $[n]$. Furthermore,
\[
    1 = f^p([n])
    =
    f^p(S_1) + \cdots  + f^p(S_N) + f^p(T)
    =
    N + f^p(T)
    = 1 + f^p(T).
\]
    The last equation holds as $N \equiv 1 \bmod m$, and addition is done in $G \cong \mathbb{Z}_m$. Thus $f^p(T) = 0$.

    Let $\pi : [n] \to [N + 1]$ be a function that takes $x \in S_i$ to $i$ and $x \in T$ to $N+1$. Consider the polymorphism $f^\pi$. Since $f = h \circ \fr$, by the definition of $\fr$ we can see that
    \[
    f^\pi(U_1, \ldots U_a) = h(f^p(\pi^{-1}(U_1)), \ldots, f^p(\pi^{-1}(U_a))).
    \]
    Now consider $f^p(\pi^{-1}(U))$. Note that $\pi^{-1}(U) = T_U \cup \bigcup_{i \in U \cap [N]} S_i$, where $T_U = T$ if $N+1\in U$, and $T_U = \emptyset$ otherwise. Thus
    \begin{multline*}
    f^p(\pi^{-1}(U))
    =
    f^p\left ( T_U \cup \bigcup_{i \in U \cap [N]} S_i \right)
    =
    f^p(T_U) + \sum_{i \in U \cap [N]} f^p(S_i) \\
    =
    0 + \sum_{i \in U \cap [N]} 1
    =
    | U \cap [N] | \bmod m,
    \end{multline*}
    where $|U \cap [N] | \bmod m$ is taken as an element of $\mathbb{Z}_m \cong G$. In other words,
    \[
    f^\pi(U_1, \ldots, U_a) = h( | U_1 \cap [N] | \bmod m, \ldots, |U_a \cap [N] | \bmod m).
    \]
    Suppose now that $U_1, \ldots, U_a$ are the set family representation of the
    input vector $x_1, \ldots, x_{N+1}$ (i.e.~$x_i = j$ if and only if $i \in U_j$) and consider the sum $\sum_{i = 1}^N \overline{x_i}$. The $j$-th coordinate of this sum is the number of $j$'s that appear in $x_1, \ldots, x_N$, modulo $m$, i.e.~$|U_j \cap [N]| \bmod m$. Thus we see that the equation above is equivalent to
    \[
    f^\pi(x_1, \ldots, x_{N+1}) = h\left( \sum_{i = 1}^N \overline{x_i}\right).
    \]
    This polymorphism ignores $x_{N+1}$, so we find that the function
    \[
    (x_1, \ldots, x_{N}) \mapsto h\left( \sum_{i = 1}^N \overline{x_i}\right)
    \]
    is also a polymorphism. But this is just $h_N$, which is thus a polymorphism as required.
\end{proof}

We can now prove the main theorem in this subsection.

\begin{proof}[Proof of \Cref{thm:genTractability}]
    We will show that $(\A, \B)$ admits a homomorphic sandwich $\A \to \E \to \B$, where $\E$ is a relational structure whose domain is $H$, and where each relation will be of the form (i) $c + M$ for some $c \in H^r$ and $M \lhd H^r$, or (ii) empty. By \Cref{lem:sandwich} this implies our desired conclusion. The homomorphism $\A \to \E$ will be given by the map $g(x) = \overline{x}$. The homomorphism $\E \to \B$ will be given by any function $h$ for which $f = h \circ \fr$. (Recall that such a function exists by \Cref{lem:hExists}.) We will construct $\E$ relation by relation, showing along the way that $g$ and $h$ are in fact homomorphisms.
    
    Consider some relation $R^\A$ of $\A$, of arity $r$, that corresponds to a relation $R^\B$ of $\B$, and $R^\E$ in $\E$. If $R^\A$ is empty then we can simply set $R^\E$ to be empty, and then $g$ and $h$ map tuples of $R^\A$ to tuples of $R^\E$, and then to tuples of $R^\B$ vacuously. Thus,
    suppose $t = (t_1, \ldots, t_r)$ is some tuple of $R^\A$, and let $M = M(R^\A)$. Then we set $R^\E = \overline{t} + M$; in other words, a tuple $x \in H^r$ will belong to this relation if and only if $x \equiv \overline{t} \bmod M$.
    
    We first show that $g$ maps $R^\A$ to $R^\E = \overline{t} + M$. Indeed, consider any tuple $x \in R^\A$. We know that $g(x) = \overline{x}$ by definition. Thus, $g(x) = \overline{x} = \overline{t} + (\overline{x} - \overline{t}) \in \overline{t} + M$. Thus $g$ maps $R^\A$ to $\overline{t} + M$.
    
    We now show that $h$ maps $R^\E = \overline{t} + M$ to $R^\B$. Consider any tuple $(a_1, \ldots, a_r) \in \overline{t} + M$. By \Cref{lem:matrix} there exists some matrix $X = (x_{ij})$ with $N \leq N_d$ columns and $r$ rows, where $N \equiv 1 \bmod m$, such that each column is an element of $R^\A$, and for each $i \in [r]$ we have
    \[
    \sum_{j = 1}^N \overline{x_{ij}} = a_i.
    \]
    Furthermore, by \Cref{lem:polys}, the function $h_N : A^N \to B$ given by
    \[
    h_N(x_1, \ldots, x_N) = h\left(\sum_{i = 1}^N \overline{x_i}\right)
    \]
    is a polymorphism. Note now that
    \begin{multline*}
        (h(a_1), \ldots, h(a_r))
        =
        \left(h\left(\sum_{j = 1}^N \overline{x_{1j}}\right), \ldots,
        h\left(\sum_{j = 1}^N \overline{x_{rj}}\right)\right) \\
        =
        (h_N(x_{11}, \ldots, x_{1N}),
        \ldots
        ,h_N(x_{r1}, \ldots, x_{rN}))
        \in R^\B.
    \end{multline*}
    The last inclusion holds since the tuple in question is the result of applying $h_N$, a polymorphism of $(\A, \B)$, to the rows of a matrix whose columns are elements in $R^\A$.
    
    Thus we note that $\A \to \E \to \B$ for some structure $\E$ that satisfies
    the conditions in \Cref{lem:sandwich}. In conclusion, $\PCSP(\A, \B)$
    is solved by $\AIP$ and is finitely tractable.
\end{proof}

\subsection{\texorpdfstring{Proof of \Cref{thm:genHardness}}{Proof of Theorem 43}}

In this section we will prove that $\PCSP(\A, \B)$ is \NP-hard if each polymorphism $f \in \Pol(\A, \B)$ is $k$-degenerate for some $k$ at most $N_d$, or has a hard set of size at most $N_h$.

\begin{lemma}\label{lem:genManyHard}
If $f \in \Pol(\A, \B)$ then $f$ cannot have more than $|B|^{a^2}$ disjoint hard sets.
\end{lemma}
\begin{proof}
Equivalently we show that any family $\mathcal{G}$ of more than $|B|^{a^2}$ disjoint sets contains a non-hard set.
Apply \Cref{lem:genZero} to $\mathcal{G}$ to find a nonempty subfamily $\{ G_1, \ldots \}$ such that $f^p(\bigcup_i G_i) = f^p(\emptyset)$. Thus $G_i \in \mathcal{G}$ is not a hard set.
\end{proof}

\begin{lemma}
Suppose $f \in \Pol^{(n)}(\A, \B)$ and $\pi : [n] \to [m]$. Then ${(f^\pi)}^p = f^p \circ \pi^{-1}$.
\end{lemma}
\begin{proof}
Note that $f_{ij}(S) = f(T_1, \ldots, T_a)$ where $T_j = S$, $T_i = [n] \setminus S$, and all the other inputs are $\emptyset$. Now, ${(f^\pi)}_{ij}(S) = f^\pi(T_1, \ldots, T_a) = f(\pi^{-1}(T_1), \ldots, \pi^{-1}(T_a)) = f_{ij}(\pi^{-1}(S))$, and so ${(f^\pi)}_{ij} = f_{ij} \circ \pi^{-1}$. Our conclusion follows by applying this fact for each $i, j \in A$.
\end{proof}

\begin{lemma}\label{lem:genInvImage}
Suppose $f \in \Pol^{(n)}(\A, \B)$ and $\pi : [n] \to [m]$. If $S$ is a hard set of $f$ then $\pi(S)$ is a hard set of $f^\pi$. 
\end{lemma}
\begin{proof}
We prove this by contrapositive. Suppose $\pi(S)$ is not a hard set of $f^\pi$. Then some $T \supseteq \pi(S)$ exists such that ${(f^\pi)}^p(T) = {(f^\pi)}^p(\emptyset)$. So $(f^p)(\pi^{-1}(T)) = {(f^\pi)}^p(T) = {(f^\pi)}^p(\emptyset) = (f^p)(\pi^{-1}(\emptyset)) = f^p(\emptyset)$. Thus $f^p(\pi^{-1}(T)) = f^p(\emptyset)$, and $S$ is not a hard set, as $S \subseteq \pi^{-1}(T)$.
\end{proof}

Let $\M_h$ denote the subset of $\Pol(\A, \B)$ whose polymorphisms have hard sets of size at most $N_h$. Let $\M_{x_1, \ldots, x_k}$ denote the subset of $\Pol(\A, \B)$ whose polymorphisms are $k$-degenerate, yet not $(k-1)$-degenerate, where $x_1, \ldots, x_k \in B^{A^2}$ are witnesses to this degeneracy. By assumption, and as no polymorphism is $1$-degenerate,
\begin{equation}\label{eqn2}
\Pol(\A, \B) = \M_h \cup \bigcup_{k = 2}^{N_d} \bigcup_{x_1 \in \range(f^p)} \ldots \bigcup_{x_k \in \range(f^p)} \M_{x_1, \ldots, x_k}.
\end{equation}

\begin{lemma}
There exists some assignment $I$ that takes $f \in \M_h^{(n)}$  to a subset of $[n]$ of size at most $|B|^{2a^2}$ such that, whenever $g \in \M_h^{(m)}$ and $g = f^\pi$ for some $\pi : [n] \to [m]$, we have that $\pi(I(f)) \cap I(g) \neq \emptyset$.
\end{lemma}
\begin{proof}
To construct $I(f)$, let $S_1, \ldots$ be a maximal sequence of disjoint hard sets of $f$ of size at most $|B|^{a^2}$, constructed greedily, and then set $I(f)$ to be the union of these sets. Since there can be at most $|B|^{a^2}$ disjoint hard sets by \Cref{lem:genManyHard}, it follows that $|I(f)| \leq |B|^{2a^2}$.

Consider now any $f, g \in \M_h$ such that $g = f^\pi$. Note that $I(f)$ contains within it a hard set $S$ of size at most $|B|^{a^2}$. Thus $\pi(I(f)) \supseteq \pi(S)$, which is a hard set of size at most $|B|^{a^2}$ by \Cref{lem:genInvImage}, and thus must intersect $I(g)$ by maximality. It follows that $\pi(I(f)) \cap I(g) \neq \emptyset$.
\end{proof}

\begin{lemma}
For $k \geq 2, x_1, \ldots, x_k \in \range(f^p)$, there exists some assignment $I$ that takes $f \in \M^{(n)}_{x_1, \ldots, x_k}$ to a subset of $[n]$ of size at most $k |B|^{a^2}$ such that, whenever $g \in \M^{(m)}_{x_1, \ldots, x_k}$ and $g = f^\pi$ for some $\pi : [n] \to [m]$ we have that $\pi(I(f)) \cap I(g) \neq \emptyset$.
\end{lemma}

\begin{proof}
To construct $I(f)$, take $S_1, \ldots, S_{k-1}$ to be disjoint sets such that $f^p(S_i) = x_i$, and take $T$ to be any set such that $f(T) = x_k$. Such sets exist since $f$ is not $(k-1)$-degenerate, and we can take all of these sets to be of size at most $|B|^{a^2}$, by replacing them with the subsets given by \Cref{lem:genSmallSet}. Let $I(f)$ be the union of $S_1, \ldots, S_{k-1}, T$.
Note that $|I(f)| \leq k |B|^{a^2}$.

Consider now any $f, g \in \M_{x_1, \ldots, x_k}$ such that $g = f^\pi$. Note that $I(f)$ contains within it disjoint sets $S_1, \ldots, S_{k-1}$ such that $f^p(S_i) = x_i$, and $I(g)$ contains within it a set $T$ such that $g^p(T) = x_k$. Now, $f^p(\pi^{-1}(T)) = {(f^\pi)}^p(T) = g^p(T) = x_k$, and thus by the $k$-degeneracy of $f$ and the disjointness of $S_1, \ldots, S_{k-1}$ it follows that $\pi^{-1}(T)$ and $S_1, \ldots, S_{k-1}$ must intersect. It follows that $\pi(I(f)) \cap I(g) \neq \emptyset$, as required.
\end{proof}

\begin{proof}[Proof of \Cref{thm:genHardness}]
We see in~\eqref{eqn2} that $\Pol(\A, \B)$ is the union of $m = 1 + \sum_{k = 2}^{N_d} {({|B|}^{a^2})}^k$ sets, each of which has an assignment $I$ that satisfies the condition of \Cref{thm:chain} if we take $C = \max(N_d|B|^{a^2}, |B|^{2a^2})$. Thus $\PCSP(\A, \B)$ is \NP-hard.
\end{proof}

\section{\texorpdfstring{$\bm{\textbf{BLP+AIP} = \textbf{AIP}}$ when $\A$ has one balanced relation}{BLP+AIP = AIP when A has one balanced relation}}

In this section we prove \Cref{thm:collapse} and \Cref{cor:transitivegroup}. 
Recall that we call a relation $R$ balanced if there exists a matrix $M$ whose
columns are tuples of $R$, where each tuple of $R$ appears as a column (possibly
a multiple times), and where the rows of $M$ are permutations of each other.

\thmblpaip*

Suppose that $A = [a]$, and the relation of $\A$ is $R = R^\A$. Furthermore suppose that each element in $[a]$ appears in $R$ (otherwise these elements can just be eliminated from $A$). Suppose $A \neq \emptyset, R \neq \emptyset$ (otherwise the conclusion is trivially true). We name the columns of the matrix that witness the balancedness of $R$ as $t_1, \ldots, t_N \in R$.

For any $i \in [a]$, let $\overline{i}$ be a unit vector in $\mathbb{Z}^a$;
i.e., it has a 1 at position $i$. For any tuple $(a_1, \ldots, a_r) \in A^r$, let $\overline{(a_1, \ldots, a_r)} = (\overline{a_1}, \ldots, \overline{a_r}) \in {(\mathbb{Z}^a)}^r$. Let $\overline{R} = \{ \overline{t} \mid t \in R \}\subseteq {(\mathbb{Z}^a)}^r$.
(We see the elements of $\mathbb{Z}^a$ as frequency vectors, and the elements of ${(\mathbb{Z}^a)}^r$ as tuples of frequency vectors.)
Endow all of these with additive structure. 
For any sets of vectors $A, B$, we let $A + B = \{ x + y \mid x \in A, y \in B \}$ and $A - B = \{ x - y \mid x \in A, y \in B \}$.
For any $k\in\mathbb{Z}$, we denote by $k\overline{R}$ the set of sums of $k$ vectors
from $\overline{R}$.

\begin{lemma}\label{lem:setAlgebra}
$(k + 1)\overline{R} - k\overline{R} + k \sum_i \overline{t_i} \subseteq (kN + 1) \overline{R}$.
\end{lemma}
\begin{proof}
If $x \in (k + 1)\overline{R} - k\overline{R} + k \sum_i \overline{t_i}$, it can be written as a sum of $k + 1$ vectors from $\overline{R}$, minus $k$ vectors from $\overline{R}$, plus $k$ copies of each vector $\overline{t_i}$. Since each tuple of $R$ appears among $t_1, \ldots, t_N$, the last $kN$ vectors in the sum above include at least $k$ copies of each vector in $\overline{R}$. By removing the $k$ subtracted vectors from the $k$ copies of each vector from $\overline{R}$, we find that $x$ can be written as a sum of $k + 1 - k + kN = kN + 1$ vectors from $\overline{R}$, i.e.~$x \in (kN + 1) \overline{R}$.
\end{proof}

For any $k \in \mathbb{Z}$, we define $S_k \subseteq \mathbb{Z}^a$ to be the set of sequences of integers that sum up to $k$, with non-negative coordinates.

\begin{lemma}\label{lem:altSym}
If $(\A,\B)$ has a 2-block-symmetric polymorphism $f$ of arity $2k+1$ then there exists a function $g : S_k \times S_{k+1} \rightarrow B$ such that $(g(x_1, y_1), \ldots, g(x_r, y_r)) \in R^\B$ for all $(x_1,  \ldots,  x_r) \in k\overline{R}, (y_1, \ldots , y_r) \in (k+1)\overline{R}$.
\end{lemma}
\begin{proof}
To compute $g(x, y)$, create two sequences of elements in $[a]$, of lengths $k$ and $k+1$, whose frequencies correspond to $x$ and $y$ respectively (i.e.~the sequence for $x = (x_1, \ldots, x_a)$ has $x_i$ appearances of $i$), and interleave these to create a sequence $a_1, \ldots a_{2k+1}$. Then $g(x, y) = f(a_1, \ldots, a_{2k+1})$.

To see why this function satisfies the required condition, suppose $(x_1, \ldots, x_r) \in k\overline{R}$ and $(y_1, \ldots, y_r) \in (k+1) \overline{R}$. Thus we can, by definition, create matrices $M$ and $N$, with $k$ and $k+1$ columns respectively, and $r$ rows, where each column is an element of $R$, and each row $i$ has frequencies corresponding to $x_i$ and $y_i$ respectively. Interleave the columns of these matrices to create a matrix $A$. Apply $f$ to the rows of $A$. We find that the image of row $i$ of $A$ is $g(x_i, y_i)$ by the symmetry of $f$; furthermore, the images of the rows of $A$ must form a tuple of $R^\B$, since $f$ is a polymorphism. This is the desired conclusion.
\end{proof}

\begin{lemma}\label{lem:altAlt}
Assume that there exists a function $f : (S_{k+1} - S_k) \to B$ for which that $(f(x_1),
  \ldots, f(x_r))\in~R^\B$ for any $x_1, \ldots, x_r \in S_{k+1} - S_k$ with $(x_1, \ldots, x_r) \in (k+1)\overline{R} - k \overline{R}$. Then, $\PCSP(\A, \B)$ has an alternating polymorphism of arity $2k + 1$.
\end{lemma}
\begin{proof}
If such a function exists, then
\[
g(x_1, \ldots, x_{2k+1}) = f(\overline{x_1} + \overline{x_3} + \cdots + \overline{x_{2k+1}} - \overline{x_2} - \overline{x_4} - \cdots - \overline{x_{2k}})
\]
is the required polymorphism.
\end{proof}

\begin{proof}[Proof of \Cref{thm:collapse}]
By \Cref{thm:aip}, $\AIP$ solves $\PCSP(\A,\B)$ if and only if
  $\Pol(\A,\B)$ contains alternating operations of all odd arities. 
By \Cref{thm:blpaip}, $\BLPAIP$ solves $\PCSP(\A,\B)$ if and
only if $\Pol(\A,\B)$ contains 2-block-symmetric operations of all odd
arities.
As any alternating operation is 2-block-symmetric, it follows that any PCSP solved by $\AIP$ is also solved by $\BLPAIP$.\footnote{This also directly follows from the definitions of the $\AIP$ and $\BLPAIP$ algorithms~\cite{BGWZ}.}
It suffices to show that 2-block-symmetric operations in $\Pol(\A,\B)$ imply alternating operations.

Fix some natural number $k$; we will now show that there exists an alternating
operation in $\Pol(\A, \B)$ of arity $2k + 1$. Since $\Pol(\A,
\B)$ contains a 2-block-symmetric operation of arity $2kN + 1$, let $f : S_{kN} \times S_{kN+1} \to B$ be the function given by
\Cref{lem:altSym}. We will construct the function required by
\Cref{lem:altAlt} in order to prove the existence of an alternating
polymorphism.

Consider the vector $v = \sum_i \overline{t_i} \in {(\mathbb{Z}^a)}^r$. We claim that $v$ is a constant vector. To see why this is the case, observe that one way to compute $\sum_i \overline{t_i}$ is to make $t_1, \ldots, t_N$ into the columns of a matrix, and then to compute the frequencies of each element of $[a]$ in each row. Element $i$ of $v$ is a tuple, where component $j$ is the number of appearances of $j$ in row $i$ in this matrix. But, since $t_1, \ldots, t_N$ witness the balancedness of $R$, these frequencies are equal for each row. Thus $v$ is indeed a constant vector;
suppose that $v = (c, \ldots, c)$ for some $c \in S_{N}$. Note that each element in $[a]$ appears in some tuple of $R$ by assumption, and each tuple of $R$ appears in the sum $\sum_i \overline{t_i}$. Thus each coordinate in $c \in \mathbb{Z}^a$ is at least 1.

The function we are interested in is $g : (S_{k+1} - S_k) \to B$, where
$g(x) = f\left(kc, x + kc\right)$.
First note that these inputs are legal inputs for the function $f$. To see why, note first that $c \in S_N$ and thus $kc \in S_{kN}$. Second, consider $x + kc$. As $x \in S_{k + 1} - S_k$, the components in $x$ sum up to 1. Since the components in $kc$ sum up to $kN$, it follows that the components in $x + kc$ sum up to $1 + kN$ as required. Furthermore, all the components of $x + kc$ are non-negative: each component of $x$ is at least $-k$, whereas each component of $c$ is at least 1, and thus each component of $kc$ is at least $k$. Thus $x + kc \in S_{kN + 1}$.

Why does $g$ satisfy the conditions from \Cref{lem:altAlt}? Consider any $x_1, \ldots, x_r \in S_{k+1} - S_k$ such that $(x_1, \ldots, x_r) \in (k+1) \overline{R} - k \overline{R}$. Note that
\begin{gather*}
(kc, \ldots, kc) = k(c, \ldots, c) = k \sum_i \overline{t_i} \in k N \overline{R},\\
(x_1 + kc, \ldots, x_r + kc) = (x_1, \ldots, x_r) + k(c, \ldots, c) 
\in (k+1) \overline{R} - k\overline{R} + k \sum_i \overline{t_i} \subseteq (kN + 1) \overline{R},
\end{gather*}
due to \Cref{lem:setAlgebra}. 
Thus, since $f$ satisfies the conditions in \Cref{lem:altSym}, 
\[
(g(x_1), \ldots, g(x_r)) = \left(f\left(k c, x_1 + k c\right), \ldots, f\left(k c, x_r + k c\right)\right) \in R^B.
\]
Thus $(g(x_1), \ldots, g(x_r))\in R^\B$, as required.
\end{proof}

\Cref{thm:collapse} does not generalise to structures with multiple relations (even just two), as the following examples show.

\begin{example}\label{rem:two}
  Consider a Boolean symmetric template $\A$ that has two balanced (and in fact
  even symmetric) relations,
  namely $R^\A=\{(0)\}$
  and $Q^\A=\{(0,1),(1,0),(1,1)\}$, which are unary and binary, respectively. 
Then $\CSP(\A)$ is solved by $\BLPAIP$, and indeed by $\BLP$, since the
  symmetric operation $\max(x_1, \ldots, x_n)$ is a polymorphism for any
  $n$~\cite{BBKO21}, but \emph{not} by $\AIP$. This is
  because $\A$ fails to have any alternating non-unary polymorphisms, even of
  arity 3: suppose $f(x, y, z)$ is such a polymorphism. Then $f(1, 1, 0) = f(0, 0, 0) = f(0,
  1, 1)$ as $f$ is alternating; and $f(0, 0, 0) = 0$ due to $R^\A$. However, due to $Q^\A$, $f(1, 1, 0)$ and $f(0, 1, 1)$ cannot both be 0. This contradiction implies our conclusion.
\end{example}

One cannot simply remove the balancedness condition from \Cref{thm:collapse}, as the following example shows.

\begin{example}%
\label{rem:one}
Let $\A$ be a Boolean template with relation $S^\A = \{ (0, 0, 1), (0, 1, 0), (0, 1, 1) \}$. Note that $S^\A$ is \emph{not} balanced. Then $\CSP(\A)$ is solved by $\BLPAIP$, and indeed by $\BLP$, since the symmetric operation $\max(x_1, \ldots, x_n)$ is a polymorphism for any $n$~\cite{BBKO21}, but not by $\AIP$. $\A$ fails to have any alternating polymorphism, even of arity 3, for exactly the same reason as the problem from \Cref{rem:two}. (The identities that would result from $R^\A$ in that example now result from the first component in each tuple in $S^\A$, and the identities that would result from $Q^\A$ in that example now result from the last two components in each tuple in $S^\A$.)
\end{example}

On the other hand, there are templates that are unbalanced for which $\AIP$ and $\BLPAIP$ have equivalent power, as the following example shows.

\begin{example}%
\label{rem:unbalanced}
Consider a Boolean template $\A$ that has one relation $P^\A = \{ (0, 1) \}$. Then $\CSP(\A)$ is solved by $\AIP$ and by $\BLPAIP$, since the alternating operation $x_1 + \cdots + x_{n} \bmod 2$ is a polymorphism of $\A$ for every odd $n$. This is in spite of the fact that $P^\A$ is unbalanced.
\end{example}

We now prove \Cref{cor:transitivegroup}.

\cortransitivegroup*

\begin{proof}
Let $R$ be the relation of $\A$, of arity $r$. It is sufficient to show that $R$ is balanced. Let $M$ be a matrix whose columns are the tuples of $R$. Suppose that the rows of $M$ are $r_1, \ldots, r_n$. We show that row $i$ is a permutation of row $j$, for arbitrary $i, j \in [r]$.

Represent the elements of $G$ as permutation matrices.
Let $\pi \in G$ be a permutation (matrix) that sends $i$ to $j$ (it exists by transitivity). Consider $\pi M$. Note that no two columns of $\pi M$ can be equal, since then two columns of $\pi^{-1} \pi M = M$ would be equal, which is false. Furthermore each column of $\pi M$ is a tuple of $R$, and thus a column of $M$, since $R$ is preserved by $\pi$. Thus we see that $\pi M$ can be seen as $M$ but with its columns permuted. In other words, for some permutation matrix $\sigma$, we have $\pi M = M \sigma^T$.

Now, let us look at row $j$ in $\pi M = M \sigma^T$. In $\pi M$ this is $r_i$ (since $\pi$ sends $i$ to $j$). In $M \sigma^T$ this is $r_j \sigma^T$. Thus $r_i = r_j \sigma^T$, i.e.~row $i$ of $M$ and row $j$ of $M$ are permutations of each other. We conclude that $R$ is balanced, as required.
\end{proof}

\Cref{cor:transitivegroup} applies to fewer structures than \Cref{thm:collapse}, as shown in the next example.

\begin{remark}%
\label{rem:eulerian}
Consider any digraph $\A$ with edge relation $E^\A$ that is strongly connected but not symmetric.
Then $E^\A$ is balanced (cf.~\Cref{app:lemmas}). On the other hand, the unique transitive permutation group with degree 2 (i.e.~the group containing the identity permutation and the permutation swapping two elements) does not preserve $E^\A$.
\end{remark}

\section{Conclusion}

Our first result classifies certain $\PCSP(\A, \B)$, where $\A$ and $\B$  are symmetric and $\B$ is functional,
into being either \NP-hard or solvable in polynomial time. This is the first step towards the following more general problem.

\begin{problem}
Classify the complexity of $\PCSP(\A, \B)$ for functional $\B$.
\end{problem}

Looking more specifically at the case $\PCSP(\inn{1}{3}, \B)$, we note that
our proof of \Cref{thm:dichotomy} implies that, for functional $\B$, we have that $\PCSP(\inn{1}{3}, \B)$ is tractable if and only if $\Eqn{m}{1} \to \B$ for some $m \leq |B|$, where $\Eqn{m}{1}$ is a relational structure over $\{ 0, \ldots, m-1\}$ with one ternary relation defined by $x + y + z \equiv 1 \bmod m$. By using the Chinese remainder theorem, $\Eqn{m}{1} = \Eqn{3^p}{1} \times \Eqn{q}{1}$, where $q$ is coprime to 3. Since this latter template contains a constant tuple (namely $(x, x, x)$ where $x$ is the inverse of 3 modulo $q$), we find that, for functional $\B$, $\PCSP(\inn{1}{3}, \B)$ is tractable if and only if $\Eqn{3^p}{1} \to \B$. 

Looking at non-functional templates $\PCSP(\inn{1}{3}, \B)$ that are tractable, all the examples the authors are aware of are either tractable for the same reason as a functional template is (i.e.~$\Eqn{3^p}{1} \to \B$), or because they include the not-all-equal predicate (i.e.~$\NAE \to \B$). Thus, we pose the following problem.

\begin{problem}\label{prob:2}
    Is $\PCSP(\inn{1}{3}, \B)$ tractable if and only if $\Eqn{3^p}{1} \times
    \NAE \to \B$ for some $p$?
\end{problem}

\Cref{prob:2} has a link with the problem of determining the complexity
of $\PCSP(\inn{1}{3}, \CC_k^+)$, where $\CC_k^+$ is a ternary symmetric template
on domain $[k]$ which contains tuples of the form $(1,1,2),\ldots, (k-1,k-1, k),
(k,k,1)$, as well as all tuples of three distinct elements (rainbow tuples).
Such templates are called \emph{cyclic}, with the cycle being $1 \to \cdots \to
k \to 1$. 

The link is the following: $\Eqn{3^p}{1} \times \NAE$ is a template
containing one cycle of length $2 \times 3^p$, together with certain rainbow
tuples --- in other words, $\Eqn{3^p}{1} \times \NAE \to \CC_{2 \times 3^p}^+$.
Likewise, $\Eqn{3^p}{1}$ has a cycle of length $3^p$ and some rainbow tuples,
i.e.~$\Eqn{3^p}{1} \to \CC_{3^p}^+$. That $\PCSP(\inn{1}{3}, \CC_k^+)$ is
tractable whenever $k = 3^p$ or $k = 2 \times 3^p$ was first observed in~\cite{brandts}. If \Cref{prob:2} were answered in the affirmative then
we would have that $\PCSP(\inn{1}{3}, \CC_k^+)$ is tractable if and only if $k = 3^p$ or $k = 2 \times 3^p$.
In particular, this would mean that $\PCSP(\inn{1}{3}, \CC_4^+)$ is \NP-hard, as conjectured
in~\cite{Barto21:stacs}.\footnote{Our structure $\CC_4^+$ is called
$\textbf{\v{C}}^+$ in~\cite{Barto21:stacs}.} 

Answering \Cref{prob:2} in the affirmative would resolve
\Cref{conj:aip}, i.e., $\PCSP(\inn{1}{3},\B)$ would be tractable (via $\AIP$) if and only if $\Eqn{3^p}{1} \times \NAE \to \B$. Perhaps determining whether
this equivalence is true might be easier than resolving 
\Cref{conj:aip}; thus we pose the following problem.

\begin{problem}\label{prob:3}
Is $\PCSP(\inn{1}{3}, \B)$ solved by $\AIP$ if and only if $\Eqn{3^p}{1} \times \NAE \to \B$ for some $p$?
\end{problem}   

There already exists such a characterisation for the power of $\AIP$ using an
infinite structure~\cite{BBKO21}. In particular, if we let $\mathbf{Z}$ be an
infinite structure whose domain is $\mathbb{Z}$, and with a tuple $(x, y,
z)$ in the relation if and only if $x + y + z = 1$, then $\PCSP(\inn{1}{3}, \B)$ is solved by
$\AIP$ if and only if $\mathbf{Z} \to \B$. We are interested in a finite template of this kind.

\medskip

Turning from problems to algorithms, our second result shows us that, for certain problems of the form $\PCSP(\A, \B)$ where $\B$ \emph{need not be functional}, and $\A, \B$ have one relation, $\AIP$ and $\BLPAIP$ have the same power. A natural question is for which other templates is it true?

\begin{problem}
    For which templates $(\A, \B)$ do $\AIP$ and $\BLPAIP$ have the same power?
\end{problem}

Equivalently~\cite{BBKO21,BGWZ}, for which templates $(\A,\B)$ does the existence of 2-block symmetric operations of all odd arities in $\Pol(\A,\B)$ imply the existence of alternating operations of all odd arities in $\Pol(\A,\B)$?

\medskip
We remark that the recent work~\cite{CKKNZ23:arxiv} does not answer any problem
from this section, and the results from~\cite{CKKNZ23:arxiv} are consistent with
positive answers to \Cref{prob:2} and \Cref{prob:3}.

{\small
\bibliographystyle{plainurl}
\bibliography{nz}

\begin{thebibliography}{10}

\bibitem{thebook}
Martin Aigner and G\"{u}nter~M. Ziegler.
\newblock {\em Proofs from THE BOOK}.
\newblock Springer Publishing Company, Incorporated, Germany, 4th edition,
  2009.

\bibitem{AB21}
Kristina Asimi and Libor Barto.
\newblock Finitely tractable promise constraint satisfaction problems.
\newblock In {\em Proc. 46th International Symposium on Mathematical
  Foundations of Computer Science (MFCS'21)}, volume 202 of {\em LIPIcs}, pages
  11:1--11:16, Dagstuhl, Germany, 2021. Schloss Dagstuhl -- Leibniz-Zentrum
  für Informatik.
\newblock \href {https://doi.org/10.4230/LIPIcs.MFCS.2021.11}
  {\path{doi:10.4230/LIPIcs.MFCS.2021.11}}.

\bibitem{Atserias22:soda}
Albert Atserias and Víctor Dalmau.
\newblock {Promise Constraint Satisfaction and Width}.
\newblock In {\em {Proc. 2022 ACM-SIAM Symposium on Discrete Algorithms
  (SODA'22)}}, pages 1129--1153, {USA}, 2022. {SIAM}.
\newblock \href {http://arxiv.org/abs/2107.05886} {\path{arXiv:2107.05886}},
  \href {https://doi.org/10.1137/1.9781611977073.48}
  {\path{doi:10.1137/1.9781611977073.48}}.

\bibitem{AGH17}
Per Austrin, Venkatesan Guruswami, and Johan H{å}stad.
\newblock (2+{$\epsilon$})-{S}at is {NP}-hard.
\newblock {\em {SIAM} J. Comput.}, 46(5):1554--1573, 2017.
\newblock URL: \url{https://eccc.weizmann.ac.il/report/2013/159/}, \href
  {https://doi.org/10.1137/15M1006507} {\path{doi:10.1137/15M1006507}}.

\bibitem{Barto21:stacs}
Libor Barto, Diego Battistelli, and Kevin~M. Berg.
\newblock {Symmetric Promise Constraint Satisfaction Problems: Beyond the
  Boolean Case}.
\newblock In {\em Proc. 38th International Symposium on Theoretical Aspects of
  Computer Science (STACS'21)}, volume 187 of {\em LIPIcs}, pages 10:1--10:16,
  Dagstuhl, Germany, 2021. Schloss Dagstuhl -- Leibniz-Zentrum f{\"u}r
  Informatik.
\newblock \href {http://arxiv.org/abs/2010.04623} {\path{arXiv:2010.04623}},
  \href {https://doi.org/10.4230/LIPIcs.STACS.2021.10}
  {\path{doi:10.4230/LIPIcs.STACS.2021.10}}.

\bibitem{BBKO21}
Libor Barto, Jakub Bulín, Andrei~A. Krokhin, and Jakub Opršal.
\newblock Algebraic approach to promise constraint satisfaction.
\newblock {\em J. {ACM}}, 68(4):28:1--28:66, 2021.
\newblock \href {http://arxiv.org/abs/1811.00970} {\path{arXiv:1811.00970}},
  \href {https://doi.org/10.1145/3457606} {\path{doi:10.1145/3457606}}.

\bibitem{Barto22:soda}
Libor Barto and Marcin Kozik.
\newblock {Combinatorial Gap Theorem and Reductions between Promise CSPs}.
\newblock In {\em Proc. 2022 ACM-SIAM Symposium on Discrete Algorithms
  (SODA'22)}, pages 1204--1220, {USA}, 2022. {SIAM}.
\newblock \href {http://arxiv.org/abs/2107.09423} {\path{arXiv:2107.09423}},
  \href {https://doi.org/10.1137/1.9781611977073.50}
  {\path{doi:10.1137/1.9781611977073.50}}.

\bibitem{BG21:sicomp}
Joshua Brakensiek and Venkatesan Guruswami.
\newblock {Promise Constraint Satisfaction: Algebraic Structure and a Symmetric
  Boolean Dichotomy}.
\newblock {\em {SIAM} J. Comput.}, 50(6):1663--1700, 2021.
\newblock \href {http://arxiv.org/abs/1704.01937} {\path{arXiv:1704.01937}},
  \href {https://doi.org/10.1137/19M128212X} {\path{doi:10.1137/19M128212X}}.

\bibitem{BG21:talg}
Joshua Brakensiek and Venkatesan Guruswami.
\newblock The quest for strong inapproximability results with perfect
  completeness.
\newblock {\em {ACM} Trans. Algorithms}, 17(3):27:1--27:35, 2021.
\newblock URL: \url{https://eccc.weizmann.ac.il/report/2017/080/}, \href
  {https://doi.org/10.1145/3459668} {\path{doi:10.1145/3459668}}.

\bibitem{BGWZ}
Joshua Brakensiek, Venkatesan Guruswami, Marcin Wrochna, and Stanislav
  {\v{Z}}ivn{\'y}.
\newblock The power of the combined basic {LP} and affine relaxation for
  promise {CSP}s.
\newblock {\em {SIAM} J. Comput.}, 49:1232--1248, 2020.
\newblock \href {http://arxiv.org/abs/1907.04383} {\path{arXiv:1907.04383}},
  \href {https://doi.org/10.1137/20M1312745} {\path{doi:10.1137/20M1312745}}.

\bibitem{brandts}
Alex Brandts.
\newblock {\em Promise Constraint Satisfaction Problems}.
\newblock PhD thesis, University of Oxford, 2022.
\newblock URL:
  \url{{https://ora.ouls.ox.ac.uk/objects/uuid:5efeff56-d5c8-42e8-9bd2-95243a1367ad}}.

\bibitem{BWZ21}
Alex Brandts, Marcin Wrochna, and Stanislav Živný.
\newblock The complexity of promise {SAT} on non-{B}oolean domains.
\newblock {\em {ACM} Trans. Comput. Theory}, 13(4):26:1--26:20, 2021.
\newblock \href {http://arxiv.org/abs/1911.09065} {\path{arXiv:1911.09065}},
  \href {https://doi.org/10.1145/3470867} {\path{doi:10.1145/3470867}}.

\bibitem{bz22:ic}
Alex Brandts and Stanislav Živný.
\newblock {Beyond PCSP(1-in-3,NAE)}.
\newblock {\em Inf. Comput.}, 289(Part A):104954, 2022.
\newblock \href {http://arxiv.org/abs/2104.12800} {\path{arXiv:2104.12800}},
  \href {https://doi.org/10.1016/j.ic.2022.104954}
  {\path{doi:10.1016/j.ic.2022.104954}}.

\bibitem{Bulatov05:classifying}
Andrei Bulatov, Peter Jeavons, and Andrei Krokhin.
\newblock Classifying the complexity of constraints using finite algebras.
\newblock {\em {SIAM} J. Comput.}, 34(3):720--742, 2005.
\newblock \href {https://doi.org/10.1137/S0097539700376676}
  {\path{doi:10.1137/S0097539700376676}}.

\bibitem{Bulatov17:focs}
Andrei~A. Bulatov.
\newblock A dichotomy theorem for nonuniform {CSP}s.
\newblock In {\em Proc. 58th Annual IEEE Symposium on Foundations of Computer
  Science (FOCS'17)}, pages 319--330, {USA}, 2017. {IEEE}.
\newblock \href {http://arxiv.org/abs/1703.03021} {\path{arXiv:1703.03021}},
  \href {https://doi.org/10.1109/FOCS.2017.37}
  {\path{doi:10.1109/FOCS.2017.37}}.

\bibitem{CKKNZ23:arxiv}
Lorenzo Ciardo, Marcin Kozik, Andrei Krokhin, Tamio-Vesa Nakajima, and
  Stanislav \v{Z}ivn\'{y}.
\newblock {1-in-3 vs. Not-All-Equal: Dichotomy of a broken promise}, 2024.
\newblock \href {http://arxiv.org/abs/2302.03456v2}
  {\path{arXiv:2302.03456v2}}.

\bibitem{cz23sicomp:clap}
Lorenzo Ciardo and Stanislav \v{Z}ivn\'y.
\newblock {CLAP: A New Algorithm for Promise CSPs}.
\newblock {\em {SIAM} J. Comput.}, 52(1):1--37, 2023.
\newblock \href {http://arxiv.org/abs/2107.05018} {\path{arXiv:2107.05018}},
  \href {https://doi.org/10.1137/22m1476435} {\path{doi:10.1137/22m1476435}}.

\bibitem{DRS05}
Irit Dinur, Oded Regev, and Clifford Smyth.
\newblock The hardness of 3-uniform hypergraph coloring.
\newblock {\em Comb.}, 25(5):519--535, 2005.
\newblock \href {https://doi.org/10.1007/s00493-005-0032-4}
  {\path{doi:10.1007/s00493-005-0032-4}}.

\bibitem{Feder98:monotone}
Tomás Feder and Moshe~Y. Vardi.
\newblock The {C}omputational {S}tructure of {M}onotone {M}onadic {S{N}{P}} and
  {C}onstraint {S}atisfaction: {A} {S}tudy through {D}atalog and {G}roup
  {T}heory.
\newblock {\em {SIAM} J. Comput.}, 28(1):57--104, 1998.
\newblock \href {https://doi.org/10.1137/S0097539794266766}
  {\path{doi:10.1137/S0097539794266766}}.

\bibitem{Ficak19:icalp}
Miron Ficak, Marcin Kozik, Miroslav Olšák, and Szymon Stankiewicz.
\newblock {Dichotomy for Symmetric Boolean PCSPs}.
\newblock In {\em Proc. 46th International Colloquium on Automata, Languages,
  and Programming (ICALP'19)}, volume 132, pages 57:1--57:12, Dagstuhl,
  Germany, 2019. Schloss Dagstuhl -- Leibniz-Zentrum für Informatik.
\newblock \href {http://arxiv.org/abs/1904.12424} {\path{arXiv:1904.12424}},
  \href {https://doi.org/10.4230/LIPIcs.ICALP.2019.57}
  {\path{doi:10.4230/LIPIcs.ICALP.2019.57}}.

\bibitem{fnotw:stacs}
Marek Filakovsk\'{y}, Tamio-Vesa Nakajima, Jakub Opr\v{s}al, Gianluca Tasinato,
  and Uli Wagner.
\newblock {Hardness of Linearly Ordered 4-Colouring of 3-Colourable 3-Uniform
  Hypergraphs}.
\newblock In {\em Proc. 41st International Symposium on Theoretical Aspects of
  Computer Science (STACS'24)}, volume 289 of {\em Leibniz International
  Proceedings in Informatics (LIPIcs)}, pages 34:1--34:19, 2024.
\newblock \href {http://arxiv.org/abs/2312.12981} {\path{arXiv:2312.12981}},
  \href {https://doi.org/10.4230/LIPIcs.STACS.2024.34}
  {\path{doi:10.4230/LIPIcs.STACS.2024.34}}.

\bibitem{GJ76}
M.~R. Garey and D.~S. Johnson.
\newblock The complexity of near-optimal graph coloring.
\newblock {\em J. {ACM}}, 23(1):43--49, 1976.
\newblock \href {https://doi.org/10.1145/321921.321926}
  {\path{doi:10.1145/321921.321926}}.

\bibitem{GS20:icalp}
Venkatesan Guruswami and Sai Sandeep.
\newblock {d-To-1 Hardness of Coloring 3-Colorable Graphs with {O(1)} Colors}.
\newblock In {\em Proc. 47th International Colloquium on Automata, Languages,
  and Programming (ICALP'20)}, volume 168 of {\em LIPIcs}, pages 62:1--62:12,
  Dagstuhl, Germany, 2020. Schloss Dagstuhl -- Leibniz-Zentrum für Informatik.
\newblock URL: \url{https://eccc.weizmann.ac.il/report/2019/116/}, \href
  {https://doi.org/10.4230/LIPIcs.ICALP.2020.62}
  {\path{doi:10.4230/LIPIcs.ICALP.2020.62}}.

\bibitem{Hell90:h-coloring}
Pavol Hell and Jaroslav Nešetřil.
\newblock On the {C}omplexity of {${H}$}-coloring.
\newblock {\em J. Comb. Theory, Ser. B}, 48(1):92--110, 1990.
\newblock \href {https://doi.org/10.1016/0095-8956(90)90132-J}
  {\path{doi:10.1016/0095-8956(90)90132-J}}.

\bibitem{Jeavons98:algebraic}
Peter~G. Jeavons.
\newblock On the {A}lgebraic {S}tructure of {C}ombinatorial {P}roblems.
\newblock {\em Theor. Comput. Sci.}, 200(1-2):185--204, 1998.
\newblock \href {https://doi.org/10.1016/S0304-3975(97)00230-2}
  {\path{doi:10.1016/S0304-3975(97)00230-2}}.

\bibitem{KO22:survey}
Andrei Krokhin and Jakub Opršal.
\newblock An invitation to the promise constraint satisfaction problem.
\newblock {\em ACM SIGLOG News}, 9(3):30--59, 2022.
\newblock \href {https://doi.org/10.1145/3559736.3559740}
  {\path{doi:10.1145/3559736.3559740}}.

\bibitem{KOWZ23}
Andrei~A. Krokhin, Jakub Opr\v{s}al, Marcin Wrochna, and Stanislav
  {\v{Z}}ivn{\'{y}}.
\newblock Topology and adjunction in promise constraint satisfaction.
\newblock {\em {SIAM} J. Comput.}, 52(1):38--79, 2023.
\newblock \href {http://arxiv.org/abs/2003.11351} {\path{arXiv:2003.11351}},
  \href {https://doi.org/10.1137/20m1378223} {\path{doi:10.1137/20m1378223}}.

\bibitem{NZ22:toct}
Tamio-Vesa Nakajima and Stanislav Živný.
\newblock {Linearly Ordered Colourings of Hypergraphs}.
\newblock {\em ACM Trans. Comput. Theory}, 14(3--4):12:1--12:19, 2022.
\newblock \href {http://arxiv.org/abs/2204.05628} {\path{arXiv:2204.05628}},
  \href {https://doi.org/10.1145/3570909} {\path{doi:10.1145/3570909}}.

\bibitem{nz23:lics}
Tamio-Vesa Nakajima and Stanislav Živný.
\newblock {Boolean symmetric vs. functional PCSP dichotomy}.
\newblock In {\em Proc. 38th Annual {ACM/IEEE} Symposium on Logic in Computer
  Science (LICS'23)}, pages 1--12. {IEEE}, 2023.
\newblock \href {https://doi.org/10.1109/LICS56636.2023.10175746}
  {\path{doi:10.1109/LICS56636.2023.10175746}}.

\bibitem{Schaefer78:stoc}
Thomas Schaefer.
\newblock The complexity of satisfiability problems.
\newblock In {\em Proc. 10th Annual ACM Symposium on the Theory of Computing
  (STOC'78)}, pages 216--226, USA, 1978. ACM.
\newblock \href {https://doi.org/10.1145/800133.804350}
  {\path{doi:10.1145/800133.804350}}.

\bibitem{Wrochna22}
Marcin Wrochna.
\newblock A note on hardness of promise hypergraph colouring, 2022.
\newblock \href {http://arxiv.org/abs/2205.14719} {\path{arXiv:2205.14719}}.

\bibitem{Zhuk20:jacm}
Dmitriy Zhuk.
\newblock A proof of the {CSP} dichotomy conjecture.
\newblock {\em J. {ACM}}, 67(5):30:1--30:78, 2020.
\newblock \href {http://arxiv.org/abs/1704.01914} {\path{arXiv:1704.01914}},
  \href {https://doi.org/10.1145/3402029} {\path{doi:10.1145/3402029}}.

\end{thebibliography}
}

\appendix
\section{Deferred proof}\label{app:lemmas}

\begin{lemma}
A digraph $\A$ is balanced if and only if it is a disjoint union of strongly connected components.
\end{lemma}
\begin{proof}
If the relation of $\A$ is balanced, then there exists a collection of tours of $\A$ containing every edge at least once. These exist since the multi-digraph formed by the columns of the matrix witnessing the balancedness of the relation of $\A$ is Eulerian. Thus, considering the tour that contains some edge $(u, v)$, we see that there exists a walk from $v$ to $u$. It follows that every weakly connected component of $\A$ is strongly connected, or equivalently $\A$ is the disjoint union of strongly connected components.

If $\A$ is the disjoint union of strongly connected components, then we create a matrix that witnesses the balancedness of the relation of $\A$. Consider any  edge $(u, v)$ of $\A$. Take the cycle formed by $(u, v)$ together with the path from $v$ to $u$. Add all of these edges to the matrix as columns. The resulting matrix contains each edge at least once; furthermore every vertex appears the same number of times in the first and the second row (since the edges form cycles). Thus $\A$ is balanced.
\end{proof}

\end{document}